\documentclass[sigconf, nonacm]{acmart}

\usepackage{pvldb}
\renewcommand{\abstractname}{ABSTRACT}
\renewcommand{\acksname}{ACKNOWLEDGMENTS}

\usepackage{amsmath}
\usepackage{graphicx}
\usepackage{textcomp}
\usepackage[normalem]{ulem}
\usepackage{listings}
\usepackage{courier}
\usepackage{booktabs}
\usepackage{multirow}
\usepackage{makecell}
\usepackage{svg}
\usepackage{url}
\usepackage{ifthen}
\usepackage{forest}
\usepackage{pifont}
\usepackage{adjustbox}
\usepackage[shortlabels]{enumitem}
\usepackage{soul}
\usepackage{multicol}
\usepackage{xparse}
\usepackage[utf8]{inputenc}

\usepackage{arydshln}
\usepackage{amsthm}

\definecolor{mymauve}{rgb}{0.58,0,0.82}
\definecolor{dkgreen}{rgb}{0,0.6,0}
\definecolor{browncolor}{rgb}{0.6, 0.3, 0.0}

\definecolor{lowdpcolor}{rgb}{0.9, 0.8, 0.2}
\definecolor{lightpurple}{rgb}{0.8, 0.7, 0.9}
\definecolor{lightpink}{rgb}{1.0, 0.8, 0.86}

\definecolor{advantagecolor}{rgb}{0.8, 1.0, 0.8}  % 更淡的绿色
\definecolor{disadvantagecolor}{HTML}{FFE5E5} %更淡的粉红

\definecolor{lowdpcolor}{HTML}{DDDDFF}
\definecolor{highdpcolor}{HTML}{FFCCCC}

\definecolor{lightred}{HTML}{FFCCCC}     % 淡红色
\definecolor{lightblue}{HTML}{CCCCFF}    % 淡蓝色

\definecolor{myblue}{rgb}{0,0,1}
\definecolor{verylightblue}{rgb}{0.8, 0.9, 1.0}
\sethlcolor{verylightblue}

\forestset{
  default preamble={
    for tree={
      parent anchor=south, 
      child anchor=north,
      align=center,
      edge={-latex},
      rounded corners,
      draw,
      fill=white,
      s sep=7mm,
      l sep=10mm,
      anchor=center,
      calign=center,
      align=center
    }
  }
}

\def\BibTeX{{\rm B\kern-.05em{\sc i\kern-.025em b}\kern-.08em
    T\kern-.1667em\lower.7ex\hbox{E}\kern-.125emX}}

\newtheorem*{scenario*}{\bf Targeted Scenarios}
\newtheorem{definition}{\bf Definition}
\newtheorem{lemma}{\bf Lemma}
\newtheorem{corollary}{\bf Corollary}
\newtheorem{theorem}{\bf Theorem}

\newtheorem{problem}{\bf Problem}
\useunder{\uline}{\ul}{}

\renewcommand\vldbdoi{10.14778/3849423.3849426}
\renewcommand\vldbpages{95 - 107}
\renewcommand\vldbissue{2}
\renewcommand\vldbyear{2026}
\renewcommand\vldbpagestyle{empty}
\renewcommand\vldbavailabilityurl{https://github.com/DDDDDstar/CORAL}
\renewcommand\vldbauthors{Shixin Wan, Guoyu Hu, Yifan Wu, Ke Chen, Lidan Shou}

\usepackage{subfigure}
\usepackage{algorithm}
\usepackage{algpseudocode}

\newtheorem{motivation}{\bf Motivation}

\makeatletter
\newcount\coralauthorbox
\pretocmd{\@typeset@author@bx}{%
  \global\advance\coralauthorbox by 1
  \ifnum\coralauthorbox>3
    \author@bx@wd=.48\textwidth
    \advance\author@bx@wd by -\author@bx@sep
  \fi
}{}{}
\makeatother

\begin{document}

\title{CORAL: Cross-modal Vector Retrieval via Incremental Graph Construction at Scale}

\author{Shixin Wan}
\affiliation{
  \institution{The State Key Laboratory of Blockchain and Data Security, Zhejiang University}
  \institution{Hangzhou High-Tech Zone (Binjiang) Institute of Blockchain\&Data Security}
  % \city{Hangzhou}
  % \country{China}
  \country{}
}
\email{wansx@zju.edu.cn}
% \authornotemark[2]

\author{Guoyu Hu}
\affiliation{
  \institution{National University of Singapore}
  % \country{China}
  \country{}
}
\email{guoyu.hu@u.nus.edu}

\author{Yifan Wu}
\affiliation{
  \institution{The State Key Laboratory of Blockchain and Data Security, Zhejiang University}
  \institution{Hangzhou High-Tech Zone (Binjiang) Institute of Blockchain\&Data Security}
  % \city{Hangzhou}
  % \country{China}
  \country{}
}
\email{yifan.wu@zju.edu.cn}
% \authornotemark[2]

\author{Ke Chen}
\affiliation{%
  \institution{The State Key Laboratory of Blockchain and Data Security, Zhejiang University}
  \institution{Hangzhou High-Tech Zone (Binjiang) Institute of Blockchain\&Data Security}
  % \city{Hangzhou}
  % \country{China}
  \country{}
}
\email{chenk@zju.edu.cn}
\authornote{Ke Chen is the corresponding author.}
% \authornotemark[2]

\author{Lidan Shou}
\affiliation{%
  \institution{The State Key Laboratory of Blockchain and Data Security, Zhejiang University}
  \institution{Hangzhou High-Tech Zone (Binjiang) Institute of Blockchain\&Data Security}
  % \city{Hangzhou}
  \country{}
}
\email{should@zju.edu.cn}

\begin{abstract}
Cross-modal vector retrieval is widely used in multimodal systems, such as search engines and vector databases. It typically operates in out-of-distribution (OOD) settings, where query vectors follow a distribution that differs from that of the vectors stored in the database. In such cases, conventional indexes suffer significant performance degradation, and even methods specially designed for OOD remain limited by inefficient use of query modal characteristics, restricted GPU parallelism, and inadequate support for dynamic updates. We present CORAL, a novel GPU-accelerated graph-based vector index for scalable cross-modal retrieval, featuring hierarchical memory management that spans GPU, CPU, and disk. Specifically, CORAL incrementally incorporates the characteristics of query modality and terminates index construction timely. Crucially, it introduces coverage-aware adaptive pruning to address the imbalanced coverage of the query vector's neighbors. Moreover, CORAL presents a fully neighborhood-aware projection approach to efficiently utilize GPUs for highly parallel index construction, and a targeted connectivity enhancement method to refine the index structure. Besides, CORAL also supports modal-semantics-based vector insertion and topology-repairing deletion that restore node connectivity. Experimental results demonstrate that CORAL outperforms existing methods with up to $1.6\times$ the throughput at matched recall while reducing construction time by up to 56\%. Furthermore, it exhibits remarkable resilience under dynamic updates and remains effective at the billion scale.
\end{abstract}

\maketitle

\vldbtopmatter

\begin{figure}[t]
    \centering
    \includegraphics[width=\linewidth]{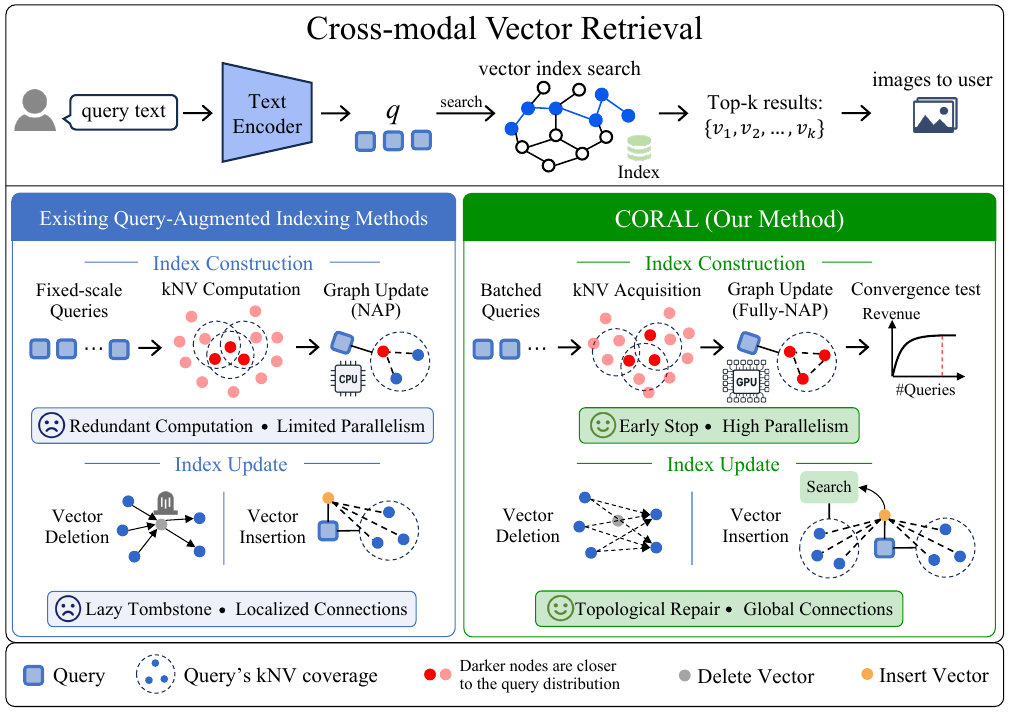}
    \caption{Comparison of CORAL and existing methods.}
    \label{fig:def}
\end{figure}

\section{INTRODUCTION}
Cross-modal vector retrieval plays a crucial role in multimodal data systems, such as search engines and vector databases~\cite{covington2016deep, karpukhin2020dense}. In these systems, unstructured data from different modalities are embedded into a shared high-dimensional vector space through advanced multimodal representation learning approaches~\cite{radford2021clip}. These high-dimensional vectors in the shared space allow queries in one modality to retrieve semantically similar instances in another, via approximate nearest neighbor search (ANNS).

The typical approach to support ANNS is to use high-dimensional vector indexes~\cite{malkov2020efficient, johnson2021billion, cao2017correlation, chen2023finger, peng2026svfusion}.
Classical indexes assume that queries and indexed data follow similar distributions and locate nearest neighbors through local geometric proximity~\cite{aumuller2023recent, baranchuk2023dedrift, dang2018image, hu2025hakes}.
Cross-modal retrieval violates this assumption because different modalities often follow distinct distributions, creating a \textit{modality gap} that reduces search effectiveness~\cite{radford2021clip, shrivastava2023theoretical}.
Recent graph-based indexes address out-of-distribution (OOD) search by adding query-guided links that better align the graph with the query distribution~\cite{chen2024roargraph, yang2025paragraph, yue2026efficient}. RoarGraph projects query neighborhoods into graph edges, ParaGraph offloads filtering to GPUs, and CDMG enforces cross-distribution monotonicity.

%第三段：介绍完最先进的工作后（第二段），总结归纳他们的挑战是什么？或者limitation是什么？
Nevertheless, the existing index methods optimized for OOD ANNS~\cite{chen2024roargraph, yang2025paragraph, yue2026efficient, jaiswal2022ood} still suffer from critical limitations:

\noindent \textbf{L1 (Inefficient handling of skewed query coverage).} 
As shown in Figure~\ref{fig:def}, \emph{query coverage} refers to the extent to which the base dataset is covered by the retrieved $k$-nearest vectors (kNV) of a set of auxiliary query vectors. These retrieved neighbors are used by existing methods to guide graph construction towards the query distribution.
Due to the modality gap, the \emph{query coverage} is highly skewed: many auxiliary queries repeatedly retrieve a small set of nodes, whereas many other base nodes are rarely retrieved and may remain isolated.
As a result, the graph is over-optimized for frequently retrieved regions but remains weakly connected elsewhere, leading to lower recall, longer traversal paths, and reduced query throughput.
Existing methods~\cite{chen2024roargraph, yang2025paragraph} mitigate such skewed coverage by introducing more auxiliary queries into the indexes being constructed. However, this strategy incurs substantial construction overhead while yielding limited structural improvement, since the additional updates remain concentrated on the already well-connected nodes.

\noindent \textbf{L2 (Limited parallelism in index update and search).}
Existing indexes are mostly designed for CPUs and remain poorly adapted to GPU execution for both graph updates and searches.
Graph updates are performed at fine granularity, where each query's kNV affects only a small number of graph nodes, resulting in limited parallelism and hindering effective batching. 
% In addition, out-degree constraints, pruning of candidate neighbors, and graph refinement introduce data dependencies and sequential traversal steps, further limiting parallel execution.
As a result, even with advanced GPU-based methods, ParaGraph's update performance remains limited~\cite{yang2025paragraph}.
In graph search, these graph-based indexes sequentially expand each node, which offers little parallelism and makes it difficult to exploit GPU acceleration~\cite{fu2022satellite, fu2019fast}.

\noindent \textbf{L3 (Poor maintenance of node connectivity for updates).}
Existing indexes rely on auxiliary queries to construct the index. As shown in the left half of Figure~\ref{fig:def}, an inserted vector (orange dot) is assigned to the coverage region of certain auxiliary queries (dashed circle) and linked with other base vectors (blue dots) contained in the same region, thereby preserving node connectivity. However, such localized updates are insufficient to maintain connectivity between inserted vectors and the global graph, especially as the query/data distributions diverge. 
For deletions, existing methods typically use a lazy tombstone mechanism~\cite{peng2023efficient}, in which obsolete nodes are marked invalid rather than explicitly removed. 
These stale nodes continue to occupy limited out-degree slots, preventing effective rewiring and causing the graph structure to gradually drift away from the updated data distribution.

% 第四部分：我们提出了什么？为了解决每一个L，我们提出了什么新技术？
To overcome these limitations, we propose CORAL, a novel graph-based index for cross-modal vector retrieval that supports efficient construction, search, and robust updates.
% incremental index construction and efficient GPU acceleration.
To address \textbf{L1}, CORAL employs an incremental construction framework that progressively enriches the index using batched query workloads. %incorporates query information.
As shown in Figure~\ref{fig:def}, queries are processed in a pipelined manner with a test mechanism that triggers early termination once information gain converges. 
In addition, we propose \textit{Coverage-Aware Adaptive Pruning} (CAAP) to optimize query coverage across the base dataset.  By promoting coverage in unexplored regions and pruning redundant updates on frequently visited nodes, our design maximizes per-query utility and eliminates extraneous computation without compromising retrieval accuracy.

To tackle \textbf{L2}, we propose a novel graph construction primitive called \textit{Fully Neighborhood-Aware Projection} (FNAP). 
Unlike traditional methods that rely on fine-grained, sequential updates, FNAP jointly leverages the structural topology within each kNV to reformulate graph evolution into coarse-grained, bulk-synchronous operations, maximizing parallelism for GPU acceleration. To handle datasets exceeding device memory, CORAL integrates a heterogeneous CPU-GPU co-processing scheme and a hierarchical index storage manager that orchestrates data movement across the GPU device, main memory, and disk storage. Complementing CAAP and FNAP, CORAL further introduces \textit{Targeted Connectivity Enhancement} (TCE), which repairs the residual isolated and over-concentrated hot nodes after coverage-aware pruning.

To address \textbf{L3}, we propose a topology-aware index update mechanism to preserve graph quality during insertions and deletions.
Insertions leverage auxiliary query vectors to integrate new data and enhance their connectivity with the entire graph through global search.
Deletions perform targeted local topology repair rather than lazy tombstone removal.
Together, these designs mitigate performance degradation over time and eliminate the need for costly periodic index reconstruction.

% 第五部分：我们的方法总体的实验效果预览，概括核心优势
We compare CORAL with state-of-the-art ANNS methods across datasets ranging from 1M to 1B in size to demonstrate the efficiency, scalability, and robustness of CORAL. 
CORAL outperforms existing methods with up to $1.6\times$ the throughput at matched recall while reducing construction time by up to 56\%. 
For index updates, CORAL sustains Recall@100 around 98\% after replacing 50\% of the original dataset, while the best competitor's recall drops below 90\% after only 10\% of the data is replaced.

%第六部分：归纳所有的贡献，要和你前面的新方案是一致的。
The main contributions are summarized as follows:
% \begin{itemize}
\begin{itemize}[leftmargin=1.5em]
\item We propose CORAL, a novel graph-based vector index for cross-modal vector retrieval. The key feature of CORAL is an incremental graph index construction approach that integrates construction and refinement to reduce redundant kNV computations and fully exploit neighborhood structural information (Section~\ref{sec:methodology}).
 
\item We propose novel topology-aware index update techniques, including vector insertion guided by modal semantics and deletion with topology repair, to preserve the quality of the graph index during runtime (Section~\ref{sec:update}). 

\item We propose high-performance graph construction and search primitives to enable CPU-GPU co-processing, and a hierarchical index storage manager for large-scale data processing across the GPU, main memory, and disk (Section~\ref{sec:gpu_implementation}).

\item We conduct extensive experiments on datasets ranging from millions to billions of vectors, showing that CORAL consistently achieves better performance than state-of-the-art OOD and in-distribution (ID) baselines in terms of recall, efficiency, and scalability. We also conduct ablation studies to validate the contribution of each component (Section~\ref{sec:experiments}).
\end{itemize}

In addition, Section~\ref{sec:prelim} provides the background on ANNS and analyzes the limitations of existing methods;
Section~\ref{sec:related_work} reviews the related work; 
and Section~\ref{sec:conclusion} concludes our work.

\section{BACKGROUND}
\label{sec:prelim}

\subsection{Out-of-Distribution ANNS}

Approximate nearest neighbor search (ANNS) is a fundamental primitive for vector retrieval.
Given a query vector, ANNS aims to retrieve its nearest vectors from a vector dataset (i.e., the base dataset) without exhaustively computing distances to all vectors (i.e., the base vectors).
In graph-based ANNS, this is typically achieved by constructing a proximity graph over the base dataset and performing greedy traversal for the query.
To formalize the retrieval target, we first define the exact nearest vectors of a query.

\begin{definition}
\textbf{$k$-nearest vectors (kNV).}
Given a query vector $q\in\mathbb{R}^d$ and a dataset $D=\{v_1,v_2,\ldots,v_N\}\subset \mathbb{R}^d$, the $k$-nearest vectors of $q$ in $D$ are the $k$ base vectors with the closest distances to $q$ under a distance function $\delta(\cdot,\cdot)$, such as Euclidean distance or inner product~\cite{cover1967nearest}. Formally,
$\mathrm{kNV}(q,D)
= \operatorname*{k\text{-}arg\,min}_{i\in[1, N]} \; \delta(q,v_i)$,
where $|\mathrm{kNV}(q,D)|=k$ and $\delta(q,v_i)\le \delta(q,v_j)$ for all $v_i\in \mathrm{kNV}(q,D)$ and $v_j\in D\setminus \mathrm{kNV}(q,D)$.
\end{definition}

The effectiveness of ANNS methods depends heavily on the relationship between the query distribution and the base distribution.
Conventional ANNS indexes are mainly designed for the in-distribution setting, where queries and base vectors follow the same distribution.
However, in cross-modal vector retrieval, queries and base vectors originate from different modalities and thus follow different distributions.
This gives rise to out-of-distribution ANNS.

\begin{definition}
\textbf{Out-of-distribution ANNS (OOD ANNS).}
Let the base dataset $D_b$ be drawn i.i.d. from a base distribution $\mathcal{P}_b$, i.e., $v_i\sim \mathcal{P}_b$.
Given a query $q$, ANNS is said to be \emph{in-distribution (ID)} if $q\sim \mathcal{P}_b$, and \emph{out-of-distribution (OOD)} if $q\sim \mathcal{P}_q$ with $\mathcal{P}_q\neq \mathcal{P}_b$.
OOD ANNS aims to retrieve an approximate result set $S$ such that $S\approx \mathrm{kNV}(q, D_b)$ without exhaustively evaluating all pairwise distances.
\end{definition}

To address OOD ANNS, existing approaches incorporate auxiliary query vectors during index construction~\cite{chen2024roargraph, yang2025paragraph, yue2026efficient}.
Graph-based indexes use these query vectors to connect their nearest base vectors, even when these base vectors are not close in the vector space.
These edges provide shortcuts for graph traversal during index search, thereby improving retrieval efficiency and accuracy.

The typical construction workflow consists of three stages.

\noindent\textbf{Stage 1: Exact kNV computation for query vectors.}
To utilize query modal information for index construction, an auxiliary query-modal vector set $Q_b$ is ingested.
For each query-modal vector $\hat{q}\in Q_b$, existing methods compute its $\mathrm{kNV}(\hat{q},D_b)$.
These base vectors in $\mathrm{kNV}(\hat{q}, D_b)$ are regarded as semantically related from the perspective of the query distribution.
The neighborhood size $k_b$ affects query coverage, i.e., the extent to which the collective kNV sets of the auxiliary queries cover the base dataset. 

\noindent\textbf{Stage 2: Neighborhood-Aware Projection (NAP).}
For each query $\hat{q}\in Q_b$, adjacency relationships are established among vectors in $\mathrm{kNV}(\hat{q}, D_b)$.
NAP selectively connects these base vectors to construct a sparse graph while controlling the maximum out-degree of each node.
Since a base vector can appear in the kNV sets of multiple queries, the projection process aims to introduce query-aware edges while preserving a bounded graph structure.

\noindent\textbf{Stage 3: Connectivity Enhancement (CE).}
After projection, the resulting graph usually still contains weakly connected or isolated components that hinder graph traversal.
CE improves graph connectivity by using sampled base vectors as queries to search on the graph and identify additional edges.
The ratio of nodes selected for CE is denoted by $\rho_{CE}$.
In traditional methods, $\rho_{CE}$ is typically set to 1, meaning that all nodes are exhaustively involved in CE.

After construction, OOD ANNS is performed by greedy graph traversal, typically implemented as beam search.
The search starts from entry nodes and iteratively expands the closest unvisited node under a bounded beam budget, returning the top-$k$ closest visited nodes upon termination.
Therefore, the effectiveness of OOD search directly depends on the structural quality of the graph.
Query-aware edges act as shortcuts to regions relevant to the query modality, allowing beam search to reach high-quality candidates with fewer distance evaluations.
In contrast, weak connectivity or isolated nodes can trap the traversal in local regions and cause recall loss.

\subsection{Problem Statement}

This paper studies the \emph{cross-modal vector retrieval} problem, which is typically addressed through OOD ANNS approaches. 
In cross-modal vector retrieval, queries and indexed data originate from different modalities (i.e., the query modality and base modality). 
A representative example is text-to-image retrieval, where textual queries are used to retrieve relevant images.
Data from the two modalities are embedded into a shared vector space to obtain query and base vectors, respectively~\cite{radford2021clip}.
We formalize the cross-modal vector retrieval problem as follows.

\begin{problem}
\textbf{Cross-modal vector retrieval.}
Let $\mathcal{X}_b$ and $\mathcal{X}_q$ denote the base and query modalities with $\mathcal{X}_b\neq \mathcal{X}_q$.
An embedding model maps data from both modalities into a shared vector space $\mathbb{R}^d$. 
Let $D_b=\{v_1,v_2,\ldots,v_N\}\subset \mathbb{R}^d$
be the base vector dataset derived from $\mathcal{X}_b$, and let $q\in\mathbb{R}^d$ be a query vector derived from $\mathcal{X}_q$.
Cross-modal vector retrieval aims to retrieve the exact top-$k$ nearest neighbors of $q$ in $D_b$ under the metric $\delta$, defined as
\[
S^{*} = \operatorname*{k\text{-}arg\,min}_{i\in[1, N]} \; \delta(q,v_i),
\]
where ties are broken arbitrarily.
Since computing $S^{*}$ exhaustively is prohibitive at scale, the goal is to return an approximate result set $S\subseteq D_b$ with $|S|=k$ that closely approximates $S^{*}$. The approximation quality is measured by $\mathrm{Recall}@k=|S\cap S^{*}|/k$, where a higher value indicates a closer approximation to the exact result. We report Recall@$k$ at $k\in\{10,100\}$ throughout the paper.
\end{problem}

\subsection{Motivation}
\label{sec:motivation}

\begin{figure}[t]
    \centering
    \subfigure[Recall of test samples vs. ingested auxiliary query size] {\label{fig:motivation-a}
        \includegraphics[width=0.3\linewidth]{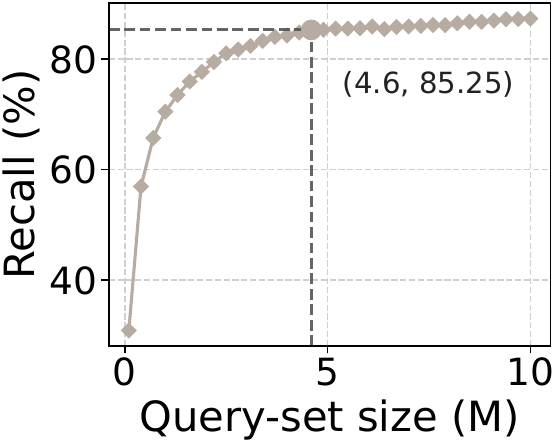}
    }
    \subfigure[Cumulative access share vs. the fraction of top-ranked base nodes] {\label{fig:motivation-b}
        \includegraphics[width=0.32\linewidth]{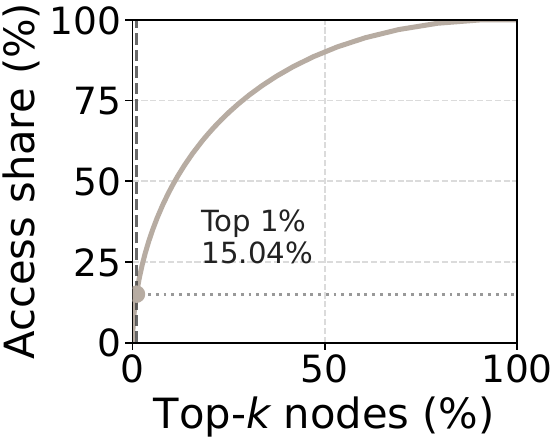}
    }
    \subfigure[Isolated-node ratio vs. ingested query size] {\label{fig:motivation-c}
        \includegraphics[width=0.3\linewidth]{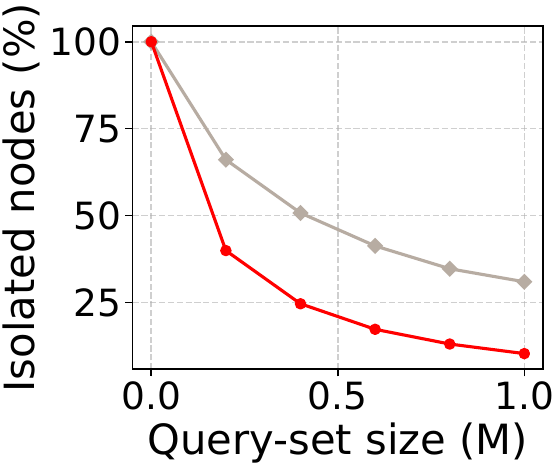}
    }
    \caption{Measurement study on T2I-10M.}
    \label{fig:motivation_study}
\end{figure}

For the cross-modal vector retrieval problem, we study existing OOD ANNS methods and conduct experimental analyses to better understand their bottlenecks.
Our observations motivate the design of a scalable and effective OOD-robust graph index.

% M1
Existing OOD ANNS methods rely on auxiliary query vectors to guide graph construction, making the construction cost highly sensitive to the size of the auxiliary query set.
However, it remains unclear how many auxiliary queries are actually necessary.
For example, RoarGraph commonly uses an auxiliary query set comparable in size to the base dataset, while experiments show that reducing the auxiliary query set by 50\% still yields comparable index performance.
As shown in Figure~\ref{fig:motivation-a}, the recall of test samples converges after ingesting about 4.6M queries, reaching 85.25\%, while construction cost keeps growing with query size.
This suggests that the benefit of ingesting additional query vectors diminishes once the graph structure has gradually converged. Beyond this point, new queries mainly trigger redundant updates on already well-connected nodes.
Therefore, the key challenge is not only to incorporate query modal information, but also to determine when sufficient query vectors have been ingested.
\begin{motivation}
During index construction, incrementally ingesting auxiliary query vectors to acquire structural information and terminating early once the graph converges achieves a better trade-off between index performance and construction cost.
\end{motivation}
Guided by this motivation, CORAL employs an incremental construction method that processes auxiliary queries in batches.
As shown in Figure~\ref{fig:def}, each batch contributes query-induced structural information, while a test mechanism periodically evaluates the index.
Once the evaluation metrics converge, CORAL terminates construction early, avoiding unnecessary kNV computation and graph updates while preserving retrieval performance.

% M2
While early termination reduces unnecessary query ingestion, it does not ensure that the ingested queries are effectively utilized.
As mentioned in \textbf{L1}, query coverage is highly skewed: a small subset of base vectors repeatedly appears in many queries' kNV sets, while many others are rarely covered.
In the index construction of the T2I-10M dataset, RoarGraph leaves over 30\% of nodes isolated after ingesting 1M queries, i.e., 100M retrieved-neighbor occurrences with $k=100$.
Moreover, as shown in Figure~\ref{fig:motivation-b}, access to the top  1\%  of base nodes accounts for  15.04\%  of the total access, indicating a strong coverage imbalance.
Such skew causes later queries to repeatedly update already well-connected nodes, while uncovered or weakly connected regions receive limited structural refinement.
This skew may originate from the true query distribution, embedding hubness, data imbalance, or redundancy in the construction algorithm, yet its impact on construction is the same regardless of the source, since repeatedly retrieving already well-connected nodes yields little new structural information.
This problem is further amplified by the construction methods of existing OOD ANNS graph-based indexes, in which each query's kNV is exploited in a restricted, asymmetric manner under out-degree constraints.
Therefore, the key challenge is to make kNV acquisition coverage-aware and increase the effective structural information provided by each query.
\begin{motivation}
During index construction, reducing redundant retrievals of over-covered nodes in kNV acquisition and redirecting later kNV acquisition toward under-covered regions improve graph connectivity and enhance the effective utilization of query vectors.
\end{motivation}
Guided by this motivation, CORAL introduces \emph{Coverage-Aware Adaptive Pruning}.
It records each base node's occurrence frequency in kNV and monitors the fraction of newly covered isolated nodes in each kNV during construction.
When this fraction is too low (i.e., below  10\% ), CORAL rebuilds the kNV acquisition index by removing high-frequency nodes from subsequent acquisition.
This redirects later kNV acquisition toward less-visited regions, allowing newly acquired kNV to provide more useful structural information.

% M3
Beyond initial construction, cross-modal vector retrieval systems are also expected to handle update workloads where vectors are continuously inserted and deleted.
As mentioned in \textbf{L3}, the localized insertion of existing OOD ANNS methods, which is restricted by the query coverage, is insufficient to maintain effective connectivity between inserted vectors and the global graph.
Meanwhile, a lazy tombstone mechanism~\cite{peng2023efficient} used for vector deletion causes local update stagnation, while directly removing deleted vectors destroys relevant local connectivity.
As updates accumulate, the lack of maintenance of local and global connectivity caused by these operations degrades retrieval performance.
\begin{motivation}
During index updates, maintaining both local and global connectivity for newly inserted and deleted vectors is crucial to preserving the robustness of index performance as data evolves.
\end{motivation}
Guided by this motivation, CORAL introduces a topology-aware update mechanism.
For insertion, CORAL augments localized query-aware connections with a refinement based on graph search to strengthen the inserted node's global connectivity.
For deletion, CORAL physically removes deleted nodes and repairs the affected local topology by reconnecting their neighbors.
Overall, these designs preserve local structural consistency and global connectivity, mitigating performance degradation as data evolves.

\begin{figure}[t]
    \centering
    \includegraphics[width=\linewidth]{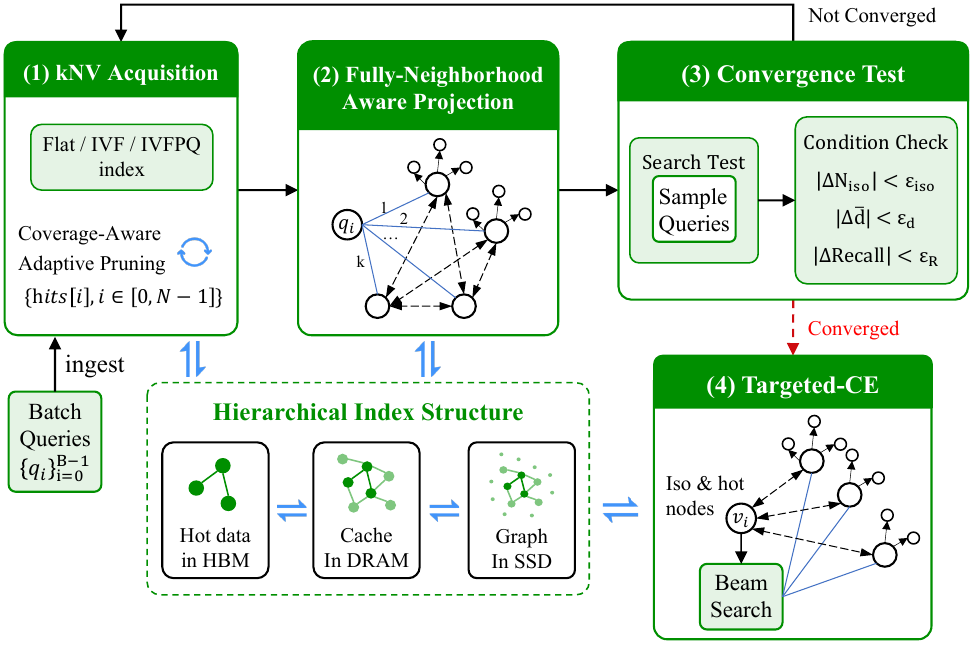}
    \caption{The workflow of CORAL's index construction.}
    \label{fig:overview}
\end{figure}

\section{CORAL}
\label{sec:methodology}

\subsection{Overview}

We present CORAL, a GPU-accelerated graph-based vector index for scalable cross-modal retrieval.
CORAL addresses three limitations of prior OOD graph methods, namely skewed query coverage, limited kNV utilization, and residual structural defects, with three techniques, Coverage-Aware Adaptive Pruning (CAAP), Fully Neighborhood-Aware Projection (FNAP), and Targeted Connectivity Enhancement (TCE), organized into two phases that share FNAP as a common graph-update primitive.

As shown in Figure~\ref{fig:def}, the first phase is an \textbf{incremental index construction} workflow with a hierarchical index structure across GPU, main memory, and disk. It ingests auxiliary queries in batches and updates the graph through kNV acquisition with CAAP, FNAP, a convergence test, and TCE (Section~\ref{sec:construction}).

The second phase is \textbf{topology-aware index update}, which runs in parallel with ongoing vector searches and includes modal-semantics-based vector insertion and topology-repairing deletion, both altering local graph structure (Section~\ref{sec:update}). Across construction, search, and update, CORAL relies on several \textbf{GPU primitives}, including FNAP and index search, that offload expensive operations such as distance computation, top-$k$ selection, pairwise distance evaluation, and batched graph updates (Section~\ref{sec:gpu_implementation}).

Finally, we provide a \textbf{theoretical analysis} of the proposed construction strategies to show that they improve graph connectivity and query-aligned navigability, thereby enhancing both index quality and search efficiency in the OOD scenario (Section~\ref{subsec:theo}).
% Finally, we detail the GPU-accelerated system optimizations that enable efficient and scalable index construction and search.

\begin{figure}[t]
    \centering
    \includegraphics[width=0.95\linewidth]{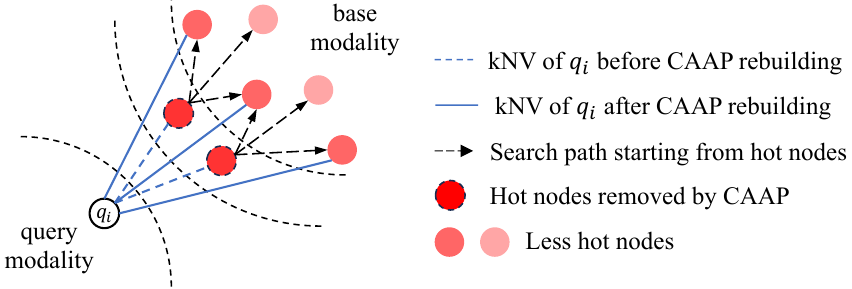}
    \caption{Illustration of CAAP and TCE.}
    \label{fig:caap}
\end{figure}

\subsection{Incremental Index Construction}
\label{sec:construction}

\noindent \textbf{Workflow.}
Figure~\ref{fig:overview} illustrates the \textbf{CORAL} index construction workflow, which consists of four modules and incrementally ingests auxiliary queries in batches:
(1) kNV acquisition for query vectors,
(2) fully neighborhood-aware projection (FNAP),
(3) convergence test, and 
(4) targeted connectivity enhancement (TCE).
% In each iteration, CORAL ingests a batch of query vectors, updates the graph, and evaluates the current index as a convergence test.

The kNV acquisition module computes the kNV of the ingested query vectors in the base dataset using indexes adapted to different data scales.
To reduce redundant graph updates due to skewed kNV coverage, CORAL applies Coverage-Aware Adaptive Pruning, which periodically rebuilds the index for kNV acquisition by removing nodes that have been excessively retrieved. 
The resulting kNV sets are then consumed by FNAP, which updates the graph by establishing adjacency between nodes within each kNV.

The convergence test module measures retrieval performance on sampled queries and monitors structural indicators such as the isolated node ratio and average out-degree. Based on these signals, CORAL decides whether further query-guided updates are still useful or whether construction should proceed to the next TCE step. 
This design avoids predefining the number of auxiliary queries and achieves a better trade-off between construction cost and index quality. 
The TCE module then mitigates the remaining structural defects by reconnecting isolated nodes and strengthening the connectivity of key routing nodes, informed by the index search.

% This approach facilitates three key mechanisms:
% coverage-aware balancing of kNV node sets, local structural enrichment via FNAP, and connectivity refinement through search-driven enhancement. These features are realized within a unified incremental graph optimization framework.
% Another advantage of CORAL is its use of GPU primitives offloading for expensive operations, including distance computation, top-$k$ selection, pairwise distance evaluation, and batched graph updates.

Moreover, to scale the index beyond the GPU's memory capacity, CORAL adopts a hierarchical index structure across the GPU, main memory, and disk.
The complete graph is stored on disk, and the main memory caches the active nodes. 
During graph updates, CORAL loads only the metadata of nodes that are necessary for distance evaluation, candidate generation, and adjacency modification.
We separate node metadata into read-only neighbor information and mutable adjacency lists. Neighbor information is suitable for sharing across computations, whereas exclusive access is enforced for adjacency-list updates to avoid conflicts.
% This design couples data movement with graph update semantics, ensuring that transferred nodes are directly consumed by the corresponding computation and avoiding redundant loading. 
By aligning data transfer with FNAP and TCE at each batch, CORAL reduces unnecessary data movement while preserving correctness, enabling a scalable index on datasets ranging from million-scale to billion-scale.

% To support large-scale deployments, the system adopts a hierarchical memory architecture for graph storage.
% Graph data are managed across GPU memory, CPU memory, and disk in a demand-driven manner.
% This design minimizes data movement and maintenance, maintains high GPU utilization, and enables seamless scaling from million–scale datasets to billion-scale datasets.

\noindent \textbf{kNV Acquisition for Datasets of Different Scales.}
Existing ANNS indexing schemes often show skewed query coverage on the base dataset in OOD scenarios.
CORAL addresses this problem with flexible kNV acquisition and coverage-aware adaptive pruning, which rely on improved use of query modal information.

When the base dataset fits in GPU memory, CORAL computes exact kNV and stores the result in GPU memory. For larger datasets, exact kNV computation becomes inefficient due to frequent host-device data transfer. 
Thus, CORAL uses approximate indexes, such as IVF~\cite{sivic2003video, babenko2015inverted} and IVFPQ~\cite{jegou2010product}, to acquire kNV. Since this stage aims to capture relatively close neighbors from the perspective of the query modality rather than achieve perfect recall, such an approximate acquisition of kNV is sufficient for indexing.

To address skewed kNV coverage, CORAL introduces Coverage-Aware Adaptive Pruning (CAAP), which periodically rebuilds the kNV search index. 
During construction, CORAL records the \emph{access frequency} of each base node, i.e., the number of times it appears in query kNV results, and monitors the fraction of newly covered isolated nodes. 
Once this fraction drops below a threshold $\tau$, indicating that newly retrieved kNV is increasingly dominated by already well-connected nodes, the kNV search index is rebuilt.

During rebuilding, the base nodes whose access frequency exceeds a threshold $t$ are identified as \emph{hot nodes}, as illustrated in Figure~\ref{fig:caap}, and are excluded from subsequent kNV acquisition, while remaining in the graph as routing nodes.
This operation directs subsequent queries toward less-frequently visited regions, so new kNV provides more useful structural information rather than repeatedly updating the already optimized regions.
The rebuild is lightweight, removing hot nodes from the candidate set for a Flat index or from the inverted lists for IVF/IVFPQ without retraining the coarse quantizer. Each construction run typically triggers 3 to 5 rebuilds, each taking about 5\,s on T2I-10M versus 635\,s total construction time.
A theoretical analysis of CAAP is provided in Section~\ref{subsec:theo}.

\begin{figure}[t]
    \centering
    \subfigure[NAP] {\label{fig:fnap-a}
        \includegraphics[width=0.22\columnwidth]{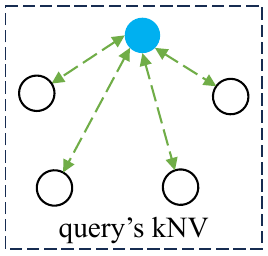}
    } 
    \subfigure[FNAP] {\label{fig:fnap-b}
        \includegraphics[width=0.33\columnwidth]{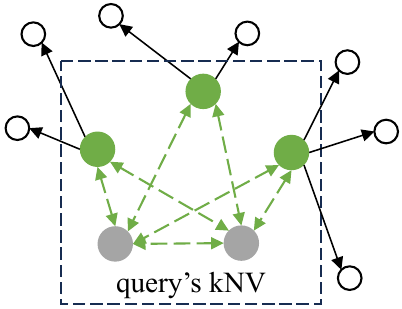}
    }
    \subfigure{
        \includegraphics[width=0.37\columnwidth]{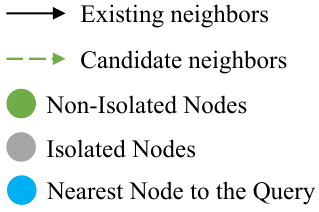}
    }
    
    \caption{Candidate-neighbor selection strategies used by NAP and FNAP.}
    \label{fig:fnap}
\end{figure}

\noindent \textbf{Fully Neighborhood-Aware Projection.}
Neighborhood-Aware Projection (NAP) is a basic step for constructing graph-based ANNS indexes under out-degree constraints. 
Given a set of nodes, NAP establishes directed edges so that each node maintains a limited number of outgoing neighbors (out-neighbors) while preserving local connectivity for beam search.

A typical instantiation is the pivot-centric NAP design, in which only the node closest to the query is selected as the pivot for each query kNV.
As shown in Figure~\ref{fig:fnap-a}, the pivot connects to the remaining $k-1$ candidate nodes, and reverse edges are added selectively. However, connections among non-pivot nodes are ignored, and several candidate nodes are pruned under the out-degree constraint without contributing to graph updates.
This design has two limitations. First, the expensive query kNV results are only partially used, leading to information loss during graph construction. 
Second, updates are centered on a single pivot, which limits parallelism and is poorly aligned with GPU execution, where updates for multiple nodes in the same kNV set could be processed in parallel.

As shown in Figure~\ref{fig:fnap-b}, FNAP generalizes NAP by allowing all nodes in a kNV to participate in adjacency updates as pivots. 
Given $\mathrm{kNV}(\hat{q}, D_b)=\{u_1,u_2,\ldots,u_k\}$ for an introduced query $\hat{q}$, FNAP treats each $u_i$ as a pivot and updates its outgoing edges using other nodes in $\mathrm{kNV}(\hat{q}, D_b)$. 
This design exploits neighborhood information more completely and exposes parallelism across pivots. 
To avoid excessive degree growth and unstable updates to existing neighborhoods, FNAP further employs a candidate selection mechanism that preserves out-degree constraints and graph stability.

FNAP categorizes nodes in $\mathrm{kNV}(\hat{q}, D_b)$ into \emph{isolated nodes}, which have no out-neighbors, and \emph{non-isolated nodes}, which already have adjacency lists.
For an isolated pivot, all other nodes in $\mathrm{kNV}(\hat{q}, D_b)$ are treated as candidates.
For a non-isolated pivot, the candidate set contains its existing neighbors together with only the isolated nodes in $\mathrm{kNV}(\hat{q}, D_b)$.
This design integrates new structural information while avoiding unnecessary updates of established edges.
After candidate selection, FNAP applies the standard NAP pruning rule to obtain the final adjacency list, which preserves locality while satisfying the out-degree constraint.
Since pivot updates are independent, FNAP naturally exposes parallelism for GPU acceleration, as detailed in Section~\ref{sec:gpu_implementation}.
Overall, FNAP serves as a unified graph-update GPU primitive for index construction and update.

\noindent \textbf{Targeted Connectivity Enhancement.}
After the graph construction converges, some base nodes remain unvisited as \emph{isolated nodes}, limiting graph connectivity and search performance. 
Existing Connectivity Enhancement methods~\cite{chen2024roargraph} usually perform exhaustive refinement on all nodes with $\rho_{CE}=1$. 
Although this improves connectivity, it also incurs high costs and generates redundant updates in already optimized regions.
To address this issue, CORAL applies TCE. It first focuses on the \emph{isolated nodes}, treats them as queries to search in the graph, and connects them with the retrieved neighbors by applying FNAP. This directly eliminates disconnected components and improves global reachability.

In addition, TCE leverages the \emph{hot nodes} identified during kNV acquisition. As shown in Figure~\ref{fig:caap}, these nodes are frequently retrieved by queries and are therefore close to the query distribution. TCE refines the neighborhoods of these nodes by updating connections to both retrieved nodes and existing neighbors. These hot nodes are also used as search entry points instead of random or centroid-based initialization~\cite{baranchuk2019learning}, thereby shortening the routing paths to query targets.
In practice, we use $\rho_{CE}$ to denote the ratio of selected \emph{isolated nodes} and \emph{hot nodes} to all nodes in the graph. A small $\rho_{CE}$ leaves the graph insufficiently connected, whereas a large one increases construction cost and over-refines regions already close to the query distribution. 
By default, $\rho_{CE}$ is set to 0.5 to achieve a better trade-off between connectivity and construction cost, as shown in the sensitivity analysis (Section~\ref{sec:parameter_sensitivity}).

TCE targets the residual structural defects that are primarily associated with \emph{isolated nodes} and \emph{hot nodes}. 
By selectively reconnecting isolated nodes and refining hot nodes, it improves graph connectivity and query-guided navigability without redundantly updating well-constructed regions.
% TCE is effective because residual structural defects are concentrated on \emph{isolated nodes} and \emph{hot nodes}. By reconnecting isolated nodes and refining hot nodes, TCE improves graph completeness and query-aligned navigability while avoiding redundant refinement of well-formed regions.

\subsection{Topology-Aware Index Update}
\label{sec:update}

\begin{figure}[t]
    \centering
    \includegraphics[width=\linewidth]{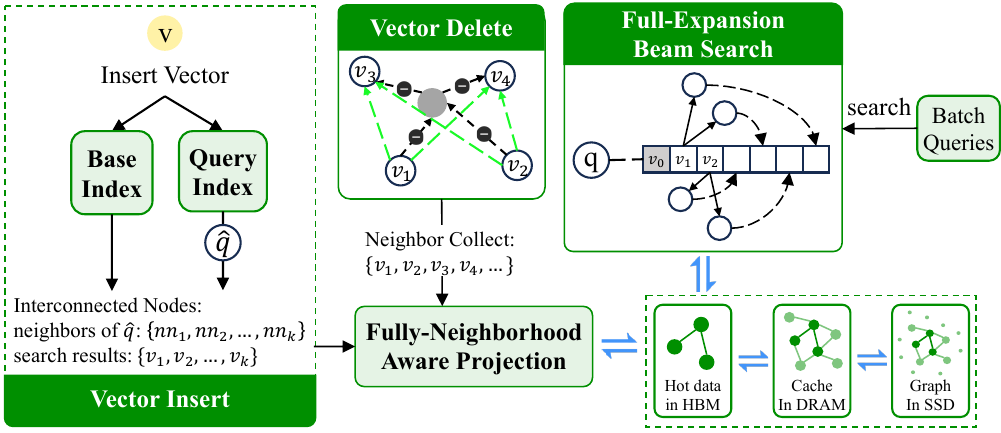}
    \caption{Overview of index search and update.}
    \label{fig:sid}
\end{figure}

% \sld{After index construction, the system supports dynamic index updates in parallel with ongoing vector searches.}
% Dynamic updates include vector insertions and deletions, both of which may alter local graph structure.

Figure~\ref{fig:sid} illustrates the interaction between search and update operations.
All graph updates, including vector insertions and deletions, are applied via the graph update primitive FNAP.

% \noindent \textbf{Vector Insertion}

The objective of vector insertion is to integrate new vectors into the index while preserving graph connectivity.
When an inserted vector $v$ is accompanied by a query vector $\hat{q}$, the system uses $\hat{q}$ to retrieve approximate nearest neighbors from the existing index.
FNAP is then applied to the node set comprising $v$ and its retrieved neighbors to establish adjacency relationships.
If no query vector is available, the system identifies an accompanying query vector by searching a Flat index built on previously ingested query vectors.
This lightweight index is retained from construction, with size $|Q_b| \ll |D_b|$ (about 0.7 GB for T2I-10M and 8\% of the graph footprint). 
It is maintained dynamically, appending query vectors upon insertions and removing queries whose retained kNV count drops below a threshold upon deletion, so its distribution evolves in sync with the base graph.
The insertion process then proceeds in the same manner as when an accompanying query vector is present.
In both cases, insertion is guided by the query modal semantics rather than solely by proximity to the base dataset's data distribution.

After insertion, CORAL enhances the new vectors' global connections similarly to TCE. Specifically, the inserted nodes are queried in the graph, and the nodes retrieved are treated as candidates for updating the neighbors of the inserted nodes.
This operation also safeguards insertion quality: even if the surrogate query is imprecise, it only provides an approximate search direction, and the graph search independently discovers global connections and corrects any inaccuracy.

% \noindent \textbf{Vector Deletion}

The objective of vector deletion is to remove nodes from the index while preserving graph connectivity.
CORAL physically removes the node to be deleted and repairs the affected graph topology.
When a node $d$ is deleted, its in-neighbors and out-neighbors are collected.
FNAP is then applied to reconstruct adjacency relationships among these neighbors, compensating for the structural disruption caused by the removal of $d$.
This repair mechanism preserves reachability and avoids structural degradation.

\subsection{GPU Primitives}
\label{sec:gpu_implementation}

This section describes how CORAL is accelerated via two key GPU primitives to achieve high throughput and scalability. 

\noindent \textbf{FNAP on the GPU.}
Given the kNV sets of a batch of $B$ queries, denoted as $\{\mathcal{K}_b\}_{b=1}^{B}$, FNAP processes each $\mathcal{K}_b$ within an independent graph update task. 
Each $\mathcal{K}_b=\{\mathbf{x}_{b,1},\ldots,\mathbf{x}_{b,k}\}\subset\mathbb{R}^d$ contains $k$ base vectors in a $d$-dimensional space, where $\mathbf{x}_{b,i}$ denotes the $i$-th vector in the $b$-th kNV. 
For each vector $\mathbf{x}_{b, i}$ treated as a pivot, FNAP takes its old neighbor set $\mathcal{N}^{\mathrm{old}}_{b, i}$ as input and produces an updated neighbor set $\mathcal{N}^{\mathrm{new}}_{b, i}$ under the maximum out-degree constraint $M$.

FNAP exposes parallelism at both the neighborhood and node granularity.
For each $\mathcal{K}_b$, pairwise distances $\delta(\mathbf{x}_{b,i},\mathbf{x}_{b,j})$ are computed in parallel within a GPU block.
Candidate generation first identifies the isolated subset 
$\mathcal{A}_b$.
For an isolated pivot $\mathbf{x}_{b, i}$, $\mathcal{C}_{b, i}$ contains all other vectors in the same kNV; otherwise, it contains the other isolated vectors and the pivot's old neighbors.
This step fuses pairwise distance computation with old-neighbor loading into a single GPU kernel, uses tiled shared-memory buffering to reuse vector data and reduce global memory access, and reuses symmetric distances to avoid redundant computation.

After candidate generation, each pivot $\mathbf{x}_{b, i}$ is processed independently by a GPU thread.
Its candidate set $\mathcal{C}_{b, i}=\{c_1,\ldots,c_m\}$, where $m=|\mathcal{C}_{b, i}|$, is sorted by increasing distance to the pivot, i.e., $\delta(\mathbf{x}_{b, i},\cdot)$. 
The algorithm scans the sorted candidates sequentially. 
If candidate $c_s$ has not been discarded, it is inserted into $\mathcal{N}^{\mathrm{new}}_{b, i}$ until the degree constraint is reached. 
Once $c_s$ is retained, later candidates $c_t$ with $t>s$ are checked in parallel and discarded if
$\delta(\mathbf{x}_{b,i},c_t)>\delta(c_s,c_t)$,
which means that $c_t$ is closer to the retained neighbor $c_s$ than to the pivot $\mathbf{x}_{b,i}$. 
This pruning rule reduces redundancy in the retained neighbor set by removing candidates that are closer to a retained neighbor than to the pivot.
Finally, graph updates are performed in batches, where each pivot independently writes back its updated neighbor set $\mathcal{N}^{\mathrm{new}}_{b, i}$. 
% A bitmap-based atomic mechanism is used to avoid redundant updates during concurrent writes, and structural statistics are aggregated in shared memory before being committed to global memory, reducing synchronization overhead. 
Overall, these optimizations enable high-throughput execution on the GPU.

\noindent \textbf{Full-Expansion Beam Search on GPU.}
Index search follows a greedy beam search strategy. 
As shown in Figure~\ref{fig:sid}, a fixed-size beam retains the current set of candidate nodes. 
At each iteration, the algorithm expands all unvisited nodes in the beam using their adjacency lists, rather than expanding only one node at a time, to expose more parallelism on the GPU. 
The newly discovered nodes are then ranked by their distances to the query and used to update the beam. 
This process repeats until the beam converges.

The selection of entry points plays a critical role in improving performance.
CORAL uses a hybrid entry strategy.
Following the TCE design, half of the entry points are selected from the hot node set obtained during kNV acquisition.
The remaining entry points are selected from the entire graph using Farthest-Point Sampling (FPS) based on hop distance, which improves coverage of regions not represented by hot nodes.
Concretely, FPS maintains for each node its minimum hop distance to the current entry set and iteratively selects the node with the largest such distance as the next entry point. This yields a diverse set of entries with global graph coverage. 

Based on these entry nodes, CORAL performs beam search using a GPU-based parallel execution strategy. Each query is assigned to a single thread block, allowing the block to maintain the beam and intermediate candidates in shared memory.
Within each block, all unexpanded nodes in the beam are expanded concurrently by each warp, and the threads in the warp evaluate their out-neighbors in parallel. 
In addition, CORAL maintains a visited bitmap to avoid expanding duplicate nodes and a dynamically updated distance threshold to guard against invalid candidate insertion, so nodes that cannot enter the beam are filtered early.
These operations align the irregular control flow of graph search with GPU execution. As a result, CORAL can explore multiple routing paths in parallel while preserving the search semantics required for efficient high-recall retrieval with limited beam queue length.

\subsection{Theoretical Analysis}
\label{subsec:theo}
We provide a theoretical justification for the proposed index construction strategies. 
Let $G_m$ denote the graph after ingesting $m$ auxiliary query vectors for construction. 
For each base node $v\in D_b$, let $h_m(v)$ be its accumulated access count in queries' kNV. 
The empirical access probability distribution over the base dataset is defined as
$p_m(v)=h_m(v)/(km)$,
where $\sum_{v\in D_b} p_m(v)=1$ since each query contributes exactly $k$ hits.
Following graph-based ANN analysis~\cite{malkov2020efficient,shrivastava2023theoretical,fu2019fast, diwan2024navigable, dong2011efficient}, we characterize the quality of the intermediate graph using the following three metrics:
\textbf{(a) Coverage Entropy:}
$\mathcal{H}(p_m)=-\sum_{v\in D_b} p_m(v)\log p_m(v)$, which measures how evenly query coverage is distributed over nodes.
\textbf{(b) Connectivity:}
$\rho(G_m)=|X_0(G_m)|/|D_b|$,
where $X_0(G_m)$ is the set of isolated nodes in $G_m$.
\textbf{(c) Routing cost:} $\mathcal{L}(q;G_m,\pi)$,
the number of visited nodes when processing query $q$ on $G_m$ under entry distribution $\pi$.

\begin{theorem}
\label{theo}
Let $\circ$ denote a graph indexing scheme for OOD ANNS, and $\star$ its enhanced scheme with the proposed CAAP+FNAP+TCE techniques, both configured with the same out-degree constraint and $m$ auxiliary query vectors.
Then, under the conditions of Corollary~\ref{coro:1}, Corollary~\ref{coro:2}, and Lemma~\ref{lemma:3}, respectively, the following holds:
$\mathcal{H}(p_m^{\star}) \ge \mathcal{H}(p_m^{\circ})$,
$\rho(G_m^{\star}) \le \rho(G_m^{\circ})$, and
$\mathcal{L}(q; G_m^{\star},\pi_H) < \mathcal{L}(q; G_m^{\circ},\pi_0)$.
\end{theorem}
We establish this theorem via the following lemmas and corollaries.

\begin{lemma}
\label{lemma:1}
Let $p$ be a probability distribution over $D_b$, and let $v_i, v_j\in D_b$ with $p_i>p_j$, where $p_i=p(v_i)$ and $p_j=p(v_j)$. Shifting probability mass $\varepsilon\in\big(0,\frac{p_i-p_j}{2}\big]$ from $p_i$ to $p_j$ strictly increases $\mathcal{H}(p)$.
\end{lemma}

\begin{proof}
The sum of the entropy terms of $v_i$ and $v_j$ with shifted probability $\varepsilon$ is
\begin{equation}
\label{equ:increase_entropy}
g(\varepsilon)
=
-\,(p_i-\varepsilon)\log(p_i-\varepsilon)
-\,(p_j+\varepsilon)\log(p_j+\varepsilon).
\end{equation}
Differentiating Equation~\eqref{equ:increase_entropy} yields
$g'(\varepsilon)=\log\frac{p_i-\varepsilon}{p_j+\varepsilon}$.
Since $p_i>p_j$ and $0<\varepsilon\le(p_i-p_j)/2$, we have 
$p_i-\varepsilon\ge p_j+\varepsilon$. 
It follows that $g'(\varepsilon)\ge 0$ on this interval, with $g'(\varepsilon)>0$ for $\varepsilon<(p_i-p_j)/2$, and hence $g(\varepsilon)>g(0)$.
The entropy terms of all other nodes are unchanged, so $\mathcal{H}(p)$ strictly increases.
\qedhere
\end{proof}

\begin{corollary}
\label{coro:1}
Assume that once $h_m(v)$ reaches $t$, CAAP excludes $v$ from subsequent kNV acquisition, so every access count is capped at $t$.
Let $p^{\circ}_m$ be the access distribution without CAAP, and let $p^{\mathrm{caap}}_m$ be the access distribution under CAAP.
If the run without CAAP yields $H_m=\{v\in D_b:\ h_m(v)>t\}\neq\emptyset$, then $\mathcal{H}(p^{\mathrm{caap}}_m)>\mathcal{H}(p^{\circ}_m)$.
\end{corollary}
\begin{proof}
Let $\bar{p}=t/(km)$; then $p^{\mathrm{caap}}_m(v)\le \bar{p}$ for every $v\in D_b$, while $p^{\circ}_m(u)>\bar{p}\ge p^{\circ}_m(v)$ for all $u\in H_m$ and $v\notin H_m$.
$p^{\mathrm{caap}}_m$ is obtained from $p^{\circ}_m$ by moving the excess mass of hot nodes from nodes above $\bar{p}$ to nodes below $\bar{p}$.
Since $H_m\neq\emptyset$, this transfer is nonempty and decomposes into a finite sequence of shifts, each moving $\varepsilon\le\min\{p_i-\bar{p},\ \bar{p}-p_j\}$ from a donor with current mass $p_i>\bar{p}$ to a receiver with current mass $p_j<\bar{p}$.
Then $p_i-\varepsilon\ge \bar{p}\ge p_j+\varepsilon$, which gives $\varepsilon\le(p_i-p_j)/2$, so every shift satisfies the condition of Lemma~\ref{lemma:1} and strictly increases the entropy.
\qedhere
\end{proof}

\begin{lemma}
\label{lemma:2}
At any step $\ell$ with kNV set $U_\ell$, let $B_\ell$ and $F_\ell$ denote the sets of isolated nodes repaired by the pivot-centric NAP and by FNAP, respectively. Then $B_\ell\subseteq F_\ell$.
\end{lemma}

\begin{proof}
Let $\deg_G(u)$ denote the out-degree of node $u$ in graph $G$, and let $I_\ell=\{u\in U_\ell:\deg_{G_{\ell-1}}(u)=0\}$ be the isolated nodes in the current kNV set.
Under the pivot-centric NAP, only the pivot $p_\ell$ and its selected neighbors $N_\ell$ receive new incident edges, where $p_\ell$ is the node closest to the query and $N_\ell$ is the neighbor set selected for $p_\ell$ after NAP pruning.
Under FNAP, every node in $U_\ell$ is processed as a pivot, and each isolated node $u\in I_\ell$ receives at least one out-neighbor from its non-empty candidate set $C_\ell(u)=U_\ell\setminus\{u\}$, so $F_\ell=I_\ell$.
Hence $B_\ell\subseteq I_\ell\cap \big(\{p_\ell\}\cup N_\ell\big)\subseteq F_\ell$.
\qedhere
\end{proof}

\begin{corollary}
\label{coro:2}
Given the same sequence of kNV sets $U_1, U_2,\ldots, U_m$ for both schemes, let $G_m^{\circ}$ be the graph built by the pivot-centric NAP, and let $G_m^{\mathrm{fnap}}$ be the graph built by FNAP.
Then $\rho(G_m^{\mathrm{fnap}})\le \rho(G_m^{\circ})$, and the inequality is strict whenever some $u\in\bigcup_\ell U_\ell$ never belongs to $\{p_\ell\}\cup N_\ell$.
\end{corollary}
\begin{proof}
Since pruning always retains at least one candidate, a node that has acquired an out-neighbor never becomes isolated again.
FNAP processes every isolated node of $U_\ell$ as a pivot with a non-empty candidate set, so $X_0(G_\ell^{\mathrm{fnap}})=X_0(G_{\ell-1}^{\mathrm{fnap}})\setminus U_\ell$, whereas Lemma~\ref{lemma:2} only gives $X_0(G_\ell^{\circ})\supseteq X_0(G_{\ell-1}^{\circ})\setminus U_\ell$.
Both runs start from $X_0=D_b$, so induction on $\ell$ yields $X_0(G_m^{\mathrm{fnap}})\subseteq X_0(G_m^{\circ})$, and dividing by $|D_b|$ gives the ratio inequality.
Unfolding the two recursions gives $X_0(G_m^{\mathrm{fnap}})=D_b\setminus\bigcup_\ell U_\ell$ and $X_0(G_m^{\circ})\supseteq D_b\setminus\bigcup_\ell(\{p_\ell\}\cup N_\ell)$, which yields the strictness condition.
\qedhere
\end{proof}

We next show that TCE, paired with hot-node entry points, achieves a lower expected routing cost than exhaustive CE with centroid-based entries.
Let $G^{\mathrm{CE}}$ be the graph obtained by exhaustive CE, and let $G^{\mathrm{TCE}}$ be the graph obtained by TCE.
For a query $q$, let $R(q)\subseteq D_b$ be its \emph{local search region}, defined as the set of the $c\ge k$ base vectors nearest to $q$, which is graph-independent.
Let $\pi_0$ be the centroid-based entry distribution, and let $\pi_H$ denote the entry distribution based on the hot-node set $H_m$.
For an entry node $s$, we define the routing cost as
\begin{equation}
\label{equ:routing_cost}
L(q;G,s)=T_{\mathrm{ent}}(q;G,s)+T_{\mathrm{loc}}(q;G),
\end{equation}
where $T_{\mathrm{ent}}(q; G,s)$ is the number of visited nodes before the search first reaches $R(q)$, i.e., the entry cost, and $T_{\mathrm{loc}}(q; G)$ is the number of additionally visited nodes after entering $R(q)$, i.e., the local traversal cost. The expected routing cost is
\begin{equation}
\label{equ:avg_routing_cost}
\mathcal{L}(q;G,\pi)=\mathbb{E}_{s\sim \pi}[L(q;G,s)].
\end{equation}
Let $F_H(r)=\Pr_{s\sim\pi_H}[d_{G^{\mathrm{TCE}}}(s,R(q))\le r]$ denote the distribution of the entry-to-region distance under the hot-node entries, and define $F_0(r)$ analogously with $\pi_0$ on $G^{\mathrm{CE}}$.
\begin{lemma}
\label{lemma:3}
Assume $R(q)$ is reachable from $\mathrm{supp}(\pi_H)$ in $G^{\mathrm{TCE}}$ and from $\mathrm{supp}(\pi_0)$ in $G^{\mathrm{CE}}$, and the following routing conditions hold.
\textit{(R1) Entry-cost monotonicity}: for each $q$ there exists a strictly increasing $\Phi_q$, shared by both graphs, with $T_{\mathrm{ent}}(q;G,s)=\Phi_q(d_G(s,R(q)))$.
\textit{(R2) Entry proximity}: $F_H(r)\ge F_0(r)$ for every $r$, strictly for some $r$.
\textit{(R3) No local-cost penalty}: $T_{\mathrm{loc}}(q;G^{\mathrm{TCE}})\le T_{\mathrm{loc}}(q;G^{\mathrm{CE}})$.
Then $\mathcal{L}(q;G^{\mathrm{TCE}},\pi_H)<\mathcal{L}(q;G^{\mathrm{CE}},\pi_0)$.
\end{lemma}

\begin{proof}
Applying \eqref{equ:routing_cost} to \eqref{equ:avg_routing_cost} gives $\mathcal{L}(q;G,\pi)\allowbreak=\mathbb{E}_{s\sim\pi}[T_{\mathrm{ent}}(q;G,s)]+T_{\mathrm{loc}}(q;G)$.
% Substituting Equation~\eqref{equ:routing_cost} into Equation~\eqref{equ:avg_routing_cost} yields $\mathcal{L}(q;G,\pi)\allowbreak=\mathbb{E}_{s\sim\pi}[T_{\mathrm{ent}}(q;G,s)]+T_{\mathrm{loc}}(q;G)$.
% Abel summation gives
% $\mathbb{E}_{s\sim\pi_0}[T_{\mathrm{ent}}(q;G^{\mathrm{CE}},s)]-\mathbb{E}_{s\sim\pi_H}[T_{\mathrm{ent}}(q;G^{\mathrm{TCE}},s)]=\sum_{r\ge0}\big(\Phi_q(r{+}1)-\Phi_q(r)\big)\big(F_H(r)-F_0(r)\big)>0$,
% since every increment of $\Phi_q$ is positive by (R1) and $F_H\ge F_0$ with strict inequality at some $r$ by (R2). Adding (R3) yields the claim.
Abel summation gives
\begin{align*}
&\mathbb{E}_{s\sim\pi_0}
  [T_{\mathrm{ent}}(q;G^{\mathrm{CE}},s)] -
  \mathbb{E}_{s\sim\pi_H}
  [T_{\mathrm{ent}}(q;G^{\mathrm{TCE}},s)] \\
&=
\sum_{r\ge 0}
\bigl(\Phi_q(r+1)-\Phi_q(r)\bigr)
\bigl(F_H(r)-F_0(r)\bigr) > 0,
\end{align*}
since every increment of $\Phi_q$ is positive by (R1) and
$F_H\ge F_0$ with strict inequality at some $r$ by (R2).
For an entry set of i.i.d.\ draws, $d_G(s,R(q))$ is replaced by its minimum over the set, which preserves the ordering in (R2), so the conclusion still holds.
\end{proof}

We next examine the three conditions, which are modeling assumptions on beam search. (R2) is the only condition involving the query distribution, and its $r=0$ case follows from the hot-node statistics.
Since $Q_b$ and the online queries are i.i.d.\ from $\mathcal{P}_q$, the ratio $h_m(v)/m$ provides an unbiased estimate of $\Pr_{q\sim\mathcal{P}_q}[v\in\mathrm{kNV}(q,D_b)]$, and $c\ge k$ gives $\mathrm{kNV}(q,D_b)\subseteq R(q)$.
Each hot node is thus contained in $R(q)$ with probability above $t/m$, whereas a uniformly sampled entry does so with probability $c/|D_b|$, a probability ratio of at least $t|D_b|/(cm)$, consistent with Figure~\ref{fig:motivation-b}, where the top 1\% most accessed nodes receive about $15\times$ the uniform access share on T2I-10M. This supports the strict inequality at $r=0$.
Finally, (R3) is corroborated by Figure~\ref{fig:diff_tce}: raising $\rho_{CE}$ from 0.5 to 1 lowers throughput by 10\%, indicating that exhaustive refinement does not reduce the local traversal cost.

Theorem~\ref{theo} follows by applying Corollary~\ref{coro:1}, Corollary~\ref{coro:2}, and Lemma~\ref{lemma:3} to its three guarantees, respectively.
When ingesting the same number of auxiliary query vectors for construction, CAAP improves the balance of query coverage, FNAP extracts more structural gain from each query's kNV, and TCE focuses refinement on nodes that dominate the search efficiency. 
Together, these designs improve graph connectivity and query-aligned navigability, thereby enhancing both index quality and search efficiency in the OOD scenario.
The empirical results of construction on T2I-10M reflect these theoretical results. As shown in Figure~\ref{fig:motivation-c}, CORAL drives the isolated-node ratio $\rho(G_m)$ below 10\% after 1M queries while RoarGraph retains over 30\%, reflecting the improved coverage balance and connectivity established by Corollary~\ref{coro:1} and Corollary~\ref{coro:2}. The routing-cost reduction of Lemma~\ref{lemma:3} is reflected in the avg-hops results in Section~\ref{sec:effectiveness}.

\begin{figure*}[t]
    \centering
    \includegraphics[width=0.95\textwidth]{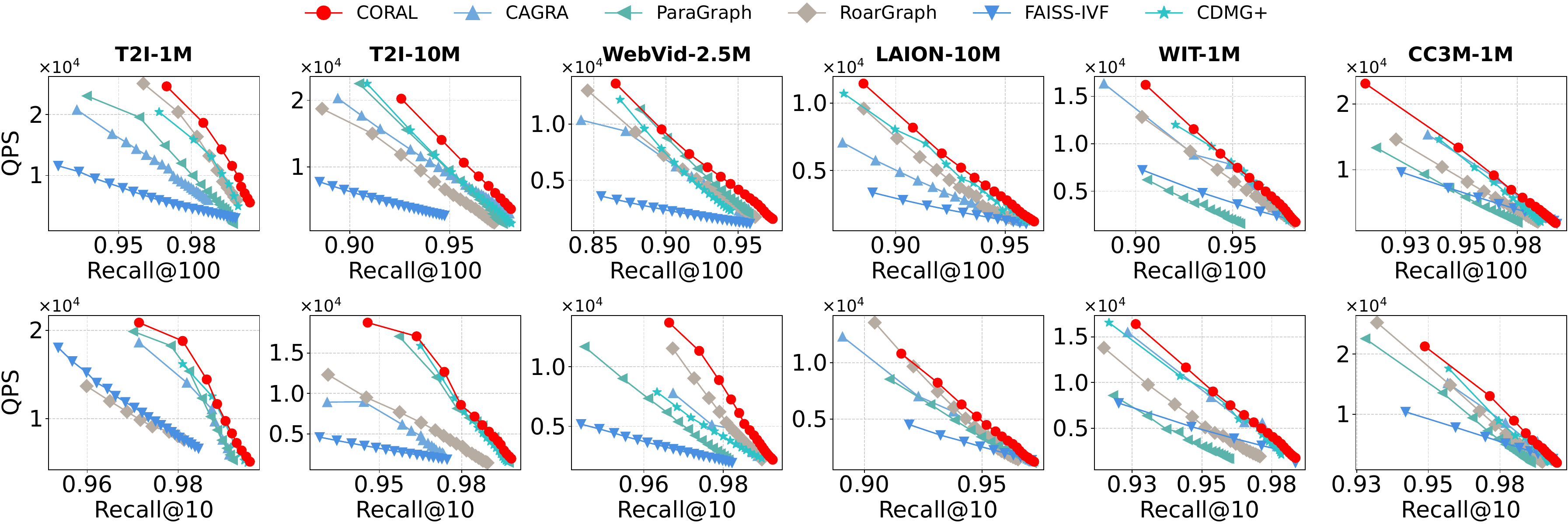}    
    \caption{Efficiency (QPS) and effectiveness (Recall@$k$) of CORAL and five other baseline algorithms across five OOD datasets: the top-right indicates better performance.}
    \label{fig:perf1}
\end{figure*}

\section{EVALUATION}
\label{sec:experiments}

We evaluate CORAL with four evaluation questions:
% \begin{description}[leftmargin=*]
% \begin{description}
\begin{description}[
  leftmargin=2.5em,
  labelwidth=2em,
  labelsep=.5em,
  align=left
]
\item[\textbf{EQ1}] \uline{Effectiveness of Cross-Modal Vector Retrieval (Section~\ref{sec:effectiveness})}:
Is CORAL more effective than existing baselines in search efficiency, scalability, and update robustness?
\item[\textbf{EQ2}] \uline{Construction Costs (Section~\ref{sec:cost_analysis})}:
Are CORAL’s construction costs competitive and practical for real-world deployment?
\item[\textbf{EQ3}] \uline{Parameter Sensitivity (Section~\ref{sec:parameter_sensitivity})}:
Are parameters $|Q_b|$, $k_b$, and $\rho_{CE}$ robust across a reasonable range and generalizable?
\item[\textbf{EQ4}] \uline{Ablation Study (Section~\ref{sec:ablation})}:
What is the contribution of each strategy and mechanism in CORAL?

\end{description}

\subsection{Experimental Setup}

\noindent \textbf{Hardware and software environment.} Unless otherwise specified, our experiments are conducted on a server running Ubuntu 24.04.1 LTS, equipped with one Intel Xeon Gold 5218 CPU, 256 GB of main memory, and one NVIDIA A100 PCIe GPU with 40 GB of HBM2 device memory. 
All GPU kernels are implemented and compiled using CUDA Toolkit 12.8.
For large-scale CPU-only experiments on the T2I-100M and T2I-1B datasets, we use an Intel Xeon Platinum 8558 CPU for index construction and search.

\begin{figure*}[t]
    \centering
    \includegraphics[width=0.95\textwidth]{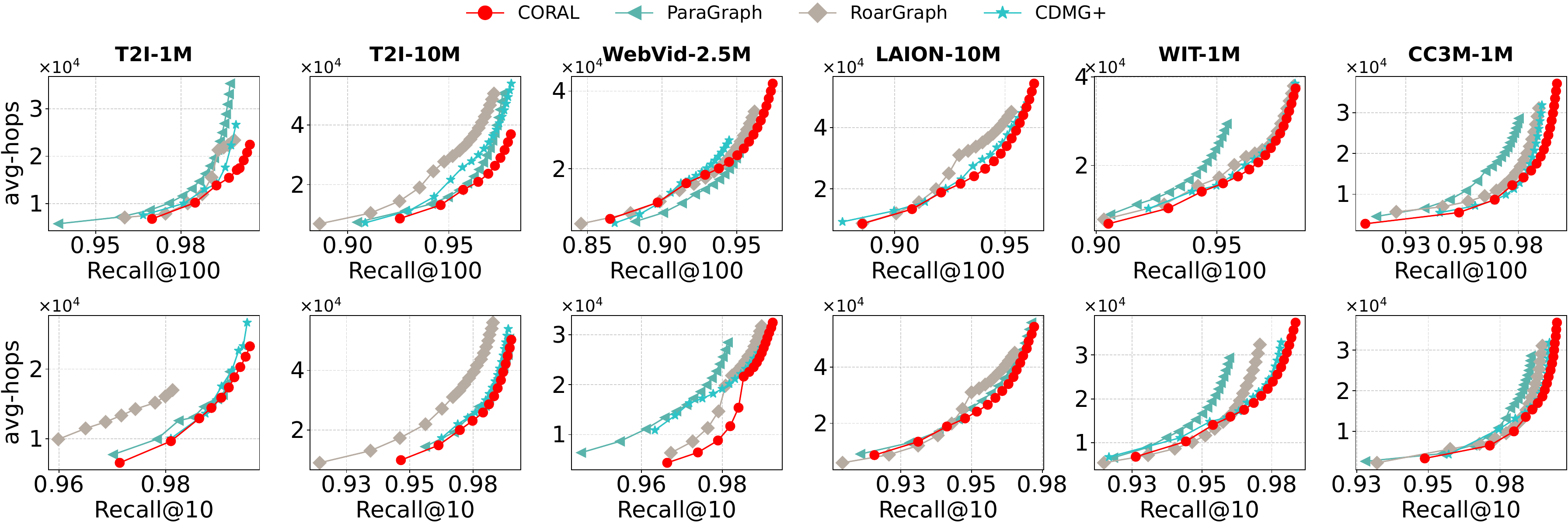}    
    \caption{Efficiency and effectiveness of CORAL and the other three baseline algorithms across five OOD datasets.}
    \label{fig:perf2}
\end{figure*}

\noindent \textbf{Datasets.} 
We use five datasets: Text-to-Image (T2I)~\cite{simhadri2022bigann}, WebVid~\cite{bain2021frozen}, LAION~\cite{schuhmann2021laion}, WIT~\cite{srinivasan2021wit}, and CC3M~\cite{sharma2018conceptual}.
T2I uses 200-dimensional vectors with $\ell_2$ distance, while WebVid, LAION, WIT, and CC3M use 512-dimensional CLIP embeddings with inner product, covering text-video and text-image retrieval scenarios.
For T2I, we construct four subsets with 1M, 10M, 100M, and 1B base vectors to evaluate scalability across different storage hierarchies.
Following CORAL's hierarchical GPU-memory/main-memory/disk index (Section~\ref{sec:construction}), the concrete layout follows the dataset scale: the 1M and 10M subsets fit entirely in GPU memory, the 100M subset fits in main memory, and the 1B subset exceeds main memory and is evaluated out-of-core with the complete graph on disk.
The corresponding auxiliary query vector scales used for construction are 0.1M, 1M, 10M, and 30M, respectively.
For the other datasets, we use 2.5M base vectors and 1M auxiliary queries for WebVid, 10M base vectors and 1M auxiliary queries for LAION, and 1M base vectors with 0.1M auxiliary queries for both WIT and CC3M.
For evaluation, we use the official 10K query vectors for T2I and randomly sample 10K queries from the original data for the remaining datasets. We ensure that they are disjoint from the auxiliary queries used for index construction.

\noindent \textbf{Baselines.} 
% We compare our approach against four OOD ANNS methods, namely RoarGraph, ParaGraph, CDMG+, and OOD-DiskANN, as well as three ID ANNS methods, namely FAISS-IVF, CAGRA, and DiskANN. 
We compare CORAL with four OOD ANNS baselines: RoarGraph, ParaGraph, CDMG+, and OOD-DiskANN.
We also include three ID ANNS baselines:
FAISS-IVF, CAGRA, and DiskANN.
The methods differ in execution profiles: CORAL, FAISS-IVF, and CAGRA construct and search on the GPU; ParaGraph constructs on the GPU but searches on the CPU; RoarGraph and CDMG+ run entirely on the CPU; and OOD-DiskANN and DiskANN are CPU-only, residing in main memory and on disk, respectively.
For fairness, for datasets with at most 10M vectors, we load the OOD baselines lacking native GPU search (RoarGraph, ParaGraph, CDMG+) into CORAL's GPU search framework, so all methods share the same search kernel and hardware, and differences reflect only the graph structure.
At the 100M and 1B scales, CORAL searches on the CPU to match OOD-DiskANN and DiskANN. At these scales, search latency is dominated by cross-tier data transfer rather than distance computation, making CPU search both fair and practical.
Below, we describe the configuration of each baseline.
% \begin{itemize}
\begin{itemize}[leftmargin=1.5em]
\item \textbf{RoarGraph}~\cite{chen2024roargraph}: a cross-modal graph-based index with neighborhood size $N = 100$ and maximum out-degree $M = 32$.
\item \textbf{ParaGraph}~\cite{yang2025paragraph}: a cross-modal graph-based index with a hybrid CPU--GPU architecture. We adopt the same experimental configuration as the original paper.
% : GPU top-1 projection graph ($degree = 15$), GPU-based search-and-refine graph ($degree = 55$), and CPU multi-round top-m projection graph ($degree = 40$). 
The query size used for index construction matches the base vector size, as required by its code.
\item  \textbf{CDMG+}~\cite{yue2026efficient}: the state-of-the-art cross-modal graph-based index, using neighborhood size $N = 100$ and maximum out-degree $M = 32$.
\item  \textbf{FAISS-IVF}~\cite{johnson2021billion}: a GPU-accelerated partition-based index for ID ANNS.
\item \textbf{CAGRA}~\cite{ootomo2024cagra}: a GPU-accelerated graph-based index for single-modal ANNS with fixed out-degree $M = 32$.
\item \textbf{OOD-DiskANN}~\cite{jaiswal2022ood}: a DiskANN variant designed for OOD queries, with $\alpha = 1.2$, max out-degree $M = 32$, and search list size $L = 64$. We use it as a baseline for the 100M-scale dataset.
\item \textbf{DiskANN}~\cite{Subramanya2019DiskANN}: a disk-based ANNS method, enabling scalable and memory-efficient search over billion-scale vector datasets, with $\alpha = 1.2$, max out-degree $M = 32$, and search list size $L = 64$ during index building. We use it as a baseline for the 1B-scale dataset because OOD-DiskANN does not provide a runnable on-disk implementation of the 1B dataset.
\end{itemize}

\begin{figure}[t]
    \centering
    \includegraphics[width=\linewidth]{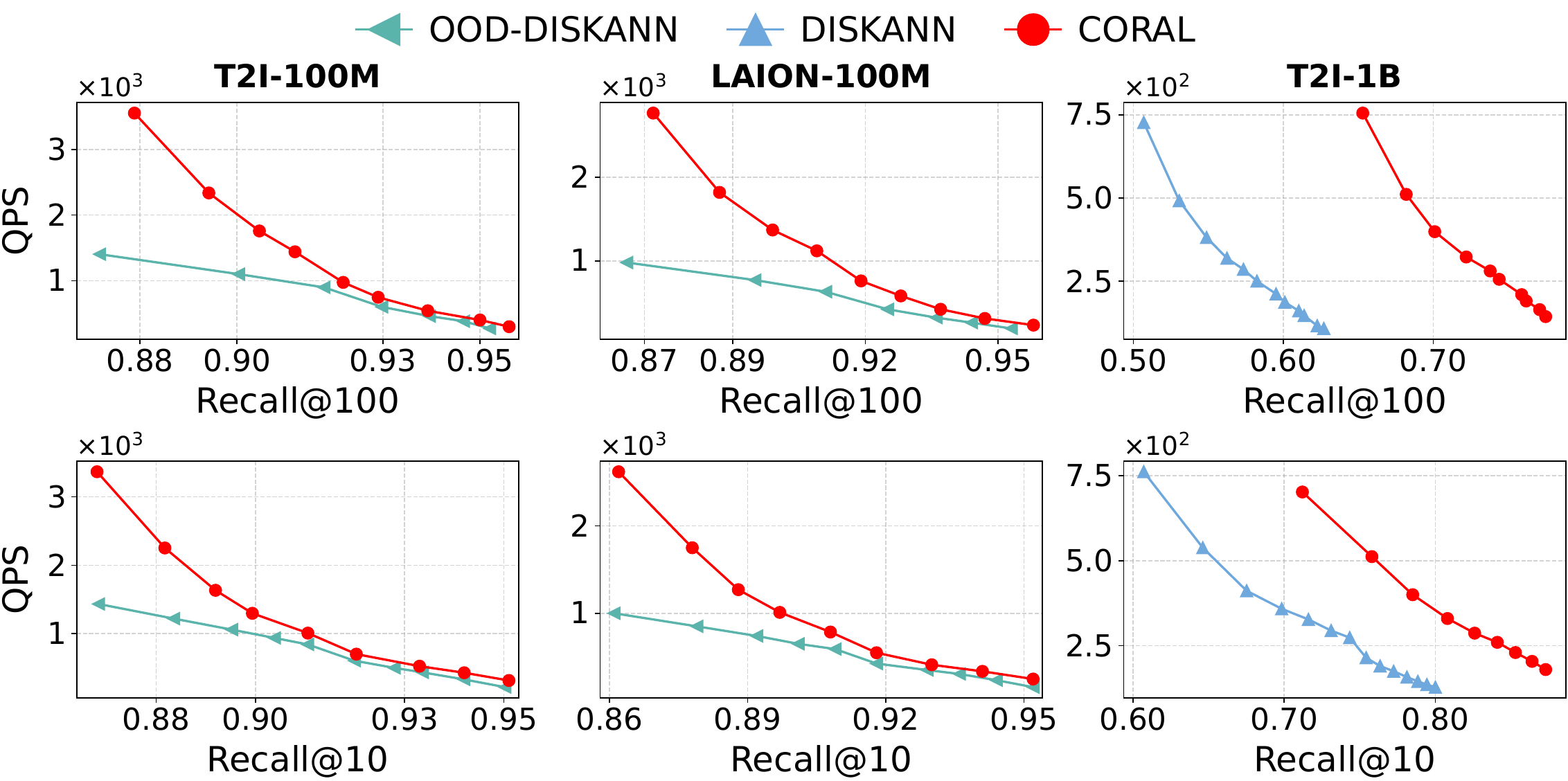}
    \caption{Effect of the GPU--CPU--disk framework of CORAL and two baseline methods across three large-scale datasets.}
    \label{fig:large}
\end{figure}

\noindent \textbf{Implementation Details.}
For CORAL's index construction, unless otherwise specified, we set $k_b$ to 100 for kNV acquisition in CAAP, the maximum out-degree to $M=32$ in FNAP, and the hot-node ratio to $\rho_{CE}=50\%$ in TCE.
These default settings balance search efficiency and construction cost, and their sensitivity is further analyzed in Section~\ref{sec:parameter_sensitivity}.

\noindent \textbf{Evaluation Metrics.}
We use Recall@$k$, the proportion of true top-$k$ nearest neighbors among the top-$k$ results, averaged over all queries, to measure retrieval accuracy in our evaluations. 
We use queries per second (QPS) to measure throughput and average hops (avg-hops) to measure computational cost during search. 
To thoroughly validate the efficacy of each method across diverse retrieval scenarios, we set $k$ to 10 and 100 for Recall@$k$. Unless otherwise stated, throughput is reported at Recall@100 = 95\%.

\subsection{Effectiveness of Vector Retrieval (EQ1)}
\label{sec:effectiveness}
We evaluate CORAL's cross-modal retrieval effectiveness in terms of search efficiency, scalability, and update robustness.

\noindent \textbf{Search Efficiency.}
Figures~\ref{fig:perf1} and \ref{fig:perf2} compare CORAL with all baselines except OOD-DiskANN and DiskANN on datasets up to 10M vectors, reporting Recall@10, Recall@100, QPS, and avg-hops. 
ParaGraph is omitted on LAION-10M because its construction exceeds GPU memory. 
Because the search procedures of FAISS-IVF and CAGRA are internally optimized, avg-hops is unavailable for these methods.
At matched recall, CORAL achieves average throughput improvements of $3.03\times$ over FAISS-IVF, $1.74\times$ over CAGRA, $1.70\times$ over RoarGraph, $2.27\times$ over ParaGraph, and $1.39\times$ over CDMG+. 
ParaGraph performs well on T2I and WebVid but poorly on WIT and CC3M, for which official configurations are unavailable and LAION-based settings are used. 

Among the graph-based OOD baselines, CORAL achieves the lowest avg-hops in most cases. 
At matched recall, it requires approximately 25\% to 36\% fewer hops than RoarGraph and ParaGraph, and approximately 20\% fewer than CDMG+, the strongest baseline in this metric.
These results are consistent with Lemma~\ref{lemma:3}.

\noindent \textbf{Index Scalability.}
Figure~\ref{fig:large} evaluates CORAL's hierarchical GPU--CPU--disk framework.
On T2I-100M, CORAL and OOD-DiskANN both store the graph in main memory.
CORAL achieves $1.39\times$ the average throughput of OOD-DiskANN at Recall@10 = 95\% and builds in 23897\,s, compared with 28809\,s for OOD-DiskANN.
The same in-memory comparison on LAION-100M shows consistently higher throughput, indicating that the advantage generalizes beyond T2I.
% On LAION-100M, evaluated against OOD-DiskANN under the same in-memory setting, CORAL exhibits the same trend, sustaining consistently higher throughput, which confirms that the scalability advantage generalizes across distinct cross-modal distributions rather than being specific to T2I.
For the T2I-1B dataset, the graph resides on disk, incurring more data transfer overhead than the 100M-scale dataset. 
At this scale, higher recall requires more disk reads per query, so throughput drops sharply as recall increases. This is an intrinsic cost of on-disk retrieval rather than a reflection of graph quality, and OOD-DiskANN itself reports below 60\% recall on T2I-1B~\cite{jaiswal2022ood}.
In this case, CORAL takes 15.5 hours to construct the index, compared with 32.3 hours for DiskANN, and achieves $2.5\times$ DiskANN's throughput at Recall@10 = 80\%, highlighting CORAL's scalability. % as an index for OOD ANNS.

\begin{figure}[t]
    \centering
    \includegraphics[width=0.95\linewidth]{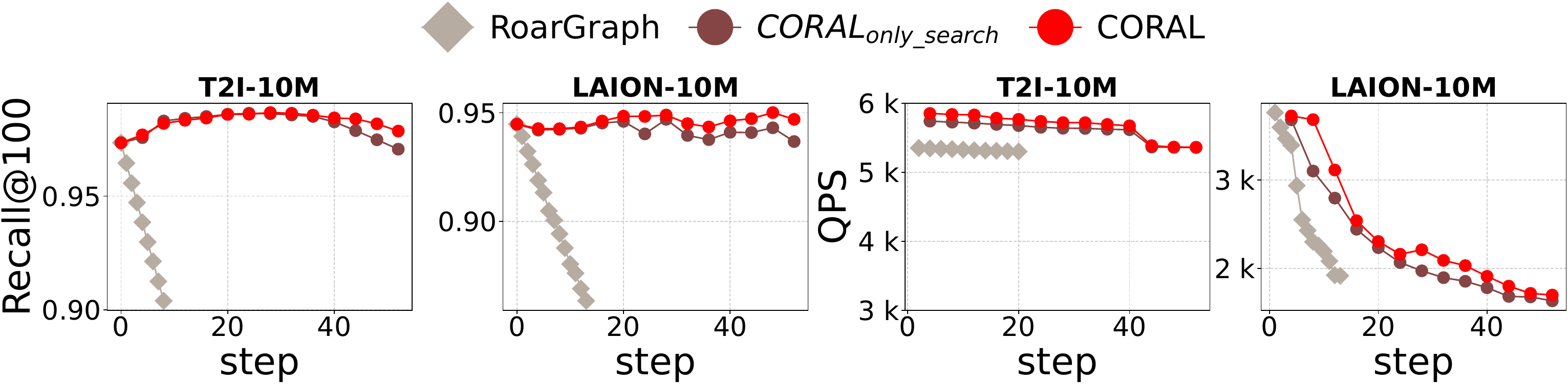}
    \caption{
    Index Recall@100 and QPS under incremental updates on T2I-10M and LAION-10M, where each step deletes and inserts 1\% of the data.
    }
    \label{fig:upd}
\end{figure}

\noindent \textbf{Update Robustness.}
Figure~\ref{fig:upd} evaluates index performance over 50 update steps. 
All methods start from the same initial graph constructed by CORAL on the first 5M vectors of each dataset, with the remaining 5M serving as the insertion pool. 
Each step deletes 1\% of the existing vectors and inserts the same number of new vectors, and all methods use the same search parameters.

RoarGraph's recall drops below 90\% within 10 steps on T2I-10M and below 87.5\% within about 10 steps on LAION-10M, while its search efficiency declines the fastest.
$\text{CORAL}_{\text{only\_search}}$ connects inserted nodes to neighbors retrieved through graph search and maintains Recall@100 above 97\% on T2I-10M after all 50 steps, although its performance begins to degrade after 40 updates. 
Full CORAL maintains Recall@100 around 98\% with largely stable throughput on T2I-10M and even improves during the early steps.
On LAION-10M, it similarly maintains approximately 95\% recall and the highest throughput throughout all 50 steps.

Across the 50 steps, CORAL sustains insertion and deletion rates of 2288 and 5367 vectors/s, respectively, compared with 1428 and 232463 vectors/s for RoarGraph. 

\noindent \textbf{Summary.} 
Overall, CORAL achieves $1.4\times$ the best baseline throughput on T2I-10M and $1.6\times$ at Recall@100 = 98\% on WebVid-2.5M, scales to 100M and 1B vectors with lower construction time, and maintains stable recall and throughput over 50 update steps. 
Its search gains stem from more effective use of query modal information and targeted structural refinement, while ID methods ignore query modal information and OOD baselines inherit limitations from RoarGraph's construction strategy. 
Under updates, RoarGraph's apparent deletion advantage comes from lazy tombstones, which accumulate invalid routing paths, while its insertion mechanism integrates poorly with the evolving graph. 
CORAL instead combines query-semantic connections with refinement from global search and repairs the graph after deletions, preventing long-term degradation.

\begin{figure}[t]
    \centering
    \includegraphics[width=\linewidth]{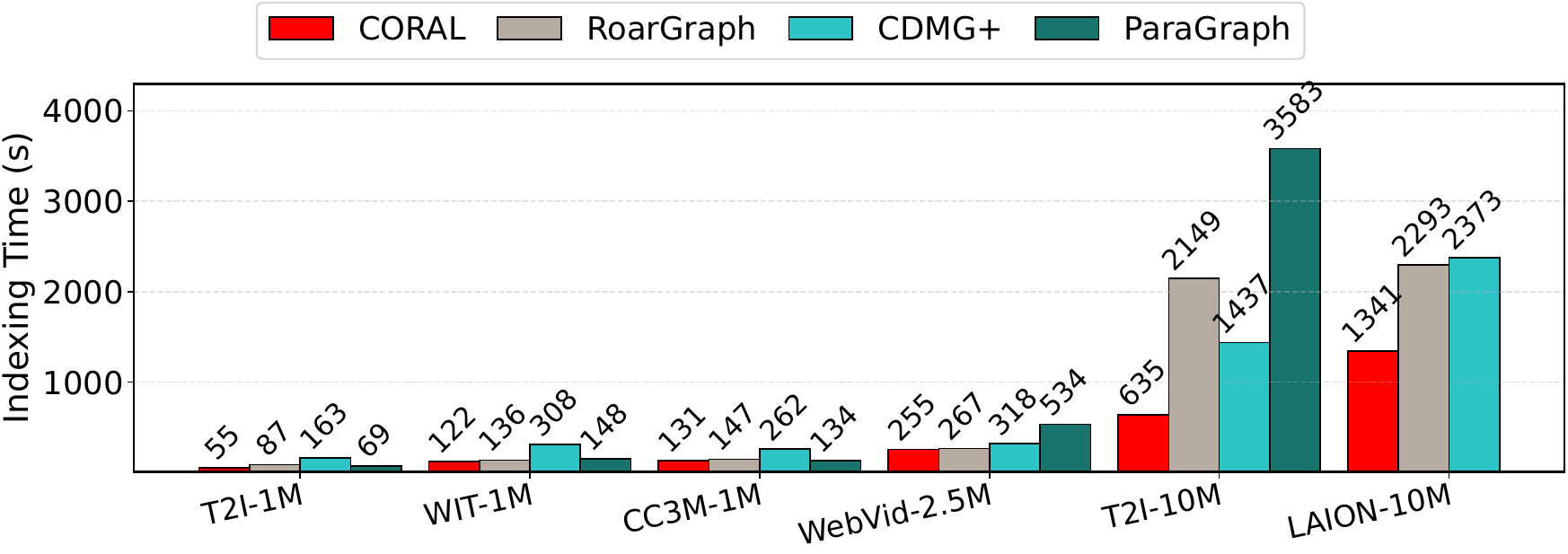}
    
    \caption{Index construction time of CORAL and three OOD ANNS baseline methods on different datasets.}
    \label{fig:perf-time}
\end{figure}

\subsection{Cost Analysis (EQ2)}
\label{sec:cost_analysis}
We analyze the storage and construction costs of CORAL.

\noindent \textbf{Storage Cost.} CORAL uses a bounded-degree graph, so its storage cost scales linearly with the dataset size and maximum out-degree, while the small set of entry point IDs adds negligible overhead.
Other bounded-degree graph baselines have similar storage costs. In contrast, FAISS-IVF stores only cluster centroids, inverted-list assignments, and optional auxiliary data, resulting in a smaller footprint.
On T2I-10M ($N=10\,\mathrm{M}$, $d=200$, $M=32$), the base vectors and adjacency lists occupy about 8.6 GB for all bounded-degree graph baselines under the same $M$.

\noindent \textbf{Construction Time.} Figure~\ref{fig:perf-time} compares index construction time across datasets. CORAL takes 635\,s and 1341\,s on T2I-10M and LAION-10M, respectively, reducing construction time by up to 56\% relative to the fastest OOD ANNS baselines.
Because CORAL and ParaGraph both construct indexes on GPUs, their comparison controls for hardware differences: on T2I-10M, CORAL takes 635\,s versus ParaGraph's 3583\,s, achieving a $5.6\times$ speedup.

\noindent \textbf{Summary.}
Overall, CORAL has storage costs comparable to other graph baselines while substantially reducing construction time.
The comparison with GPU-based ParaGraph shows that this gain stems from CORAL's algorithmic design, including GPU-parallel graph updates, efficient use of query information, and TCE.

\begin{figure}[t]
    \centering
    \subfigure[Index performance.]{\label{fig:diff}
        \includegraphics[width=0.475\linewidth]{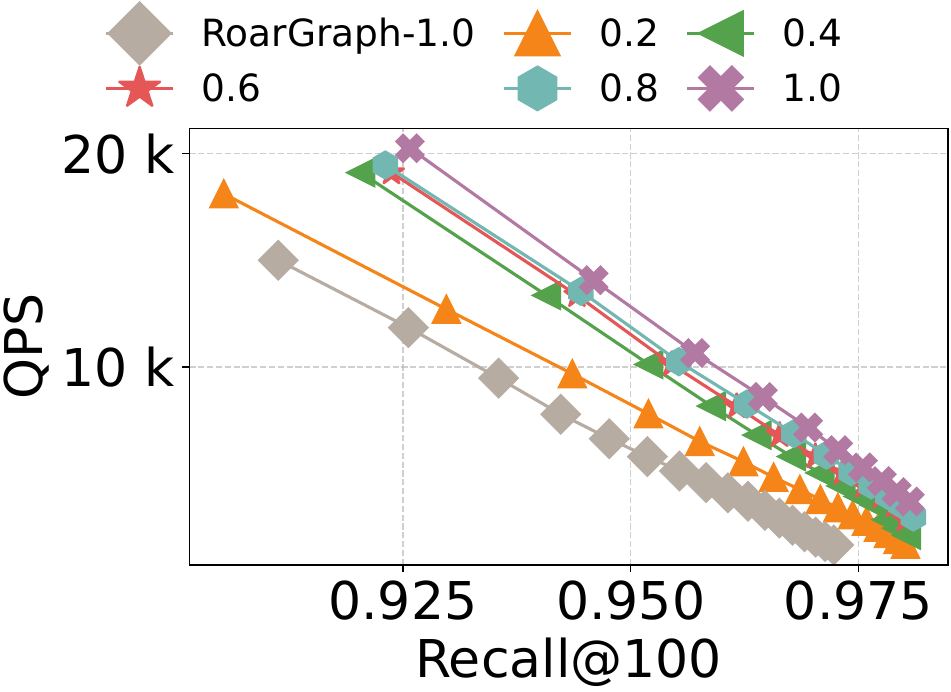}
    }
    \subfigure[Construction time.]{\label{fig:cons_time}
        \includegraphics[width=0.44\linewidth]{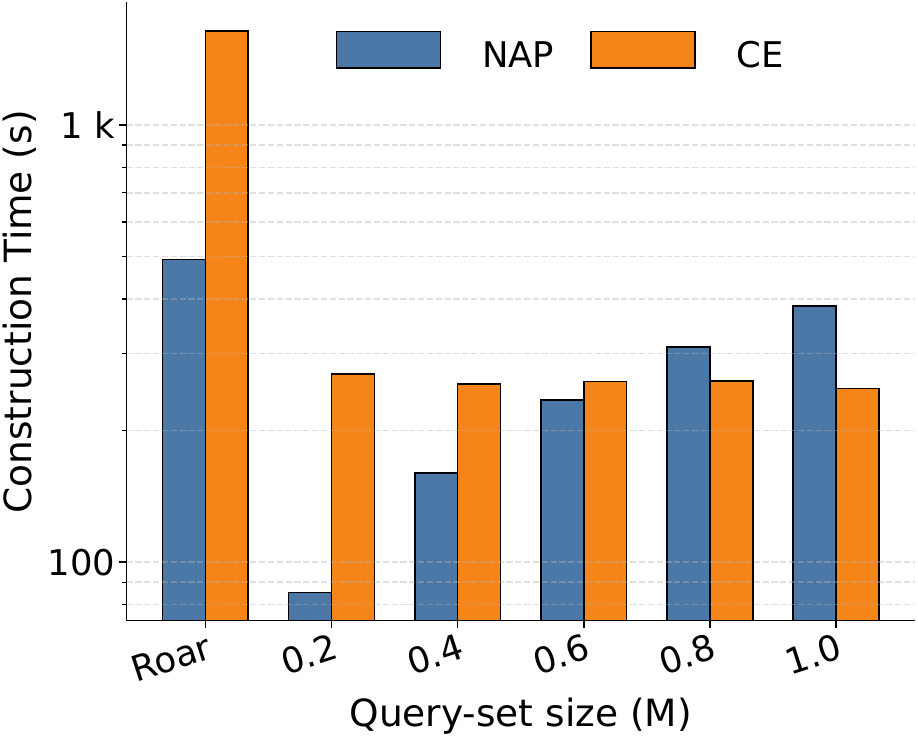}
    }
    \subfigure[Effect of $k_b$ on kNV acquisition.]{\label{fig:diff_k}
        \includegraphics[width=0.475\linewidth]{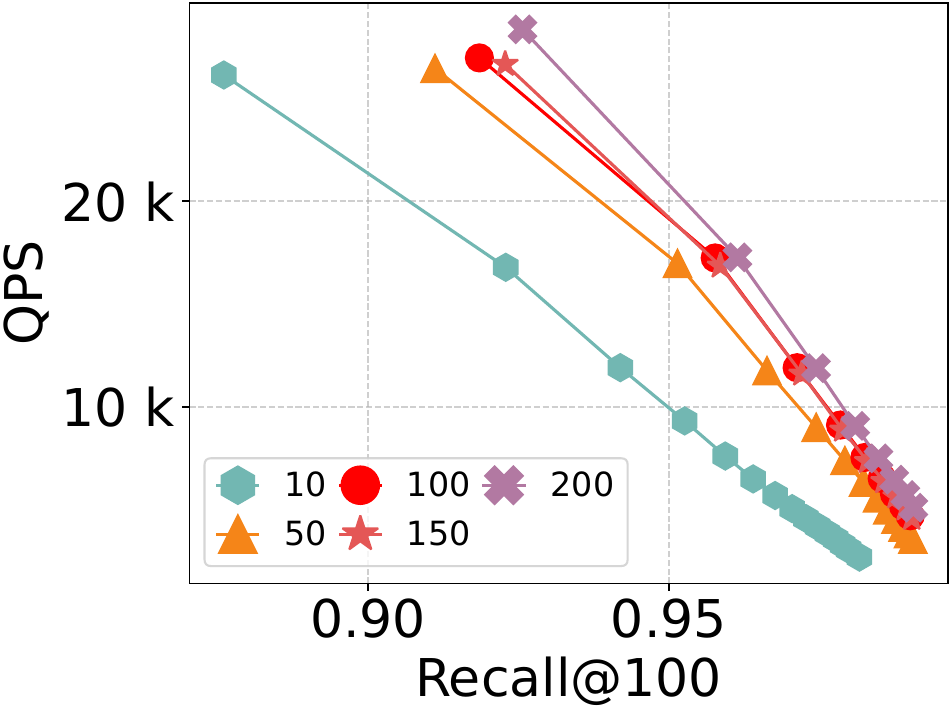}
    }
    \subfigure[Effect of $\rho_{CE}$ in TCE.]{\label{fig:diff_tce}
        \includegraphics[width=0.475\linewidth]{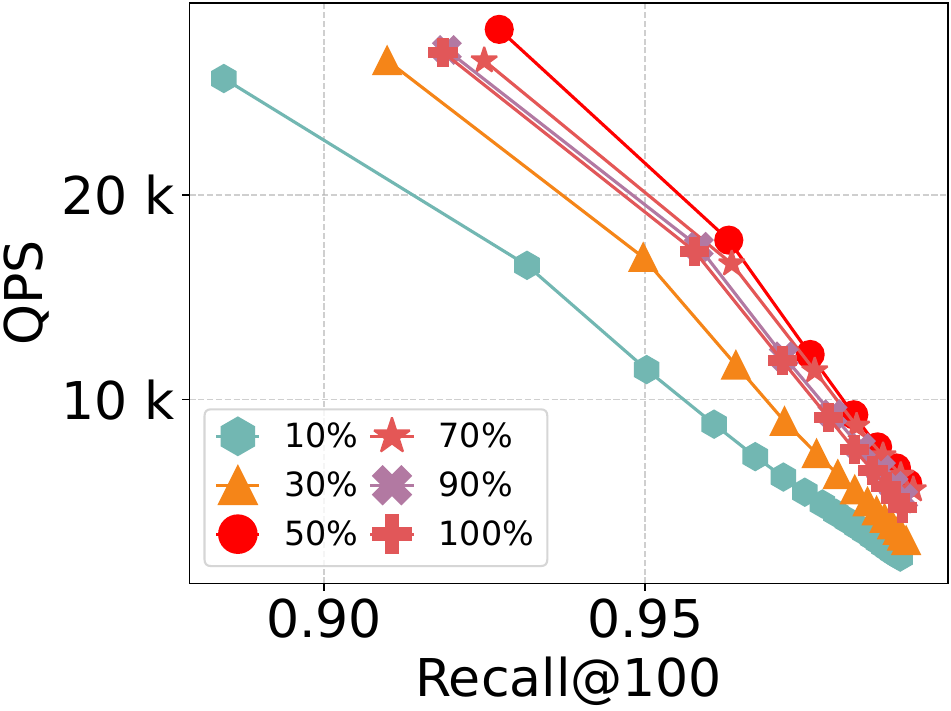}
    }
    \caption{Sensitivity analysis on $|Q_b|$, $k_b$, and $\rho_{CE}$.}
    \label{fig:diff-q}
\end{figure}

\subsection{Parameter Sensitivity Analysis (EQ3)}
\label{sec:parameter_sensitivity}

We evaluate CORAL's sensitivity to three key parameters: the number of ingested query vectors $|Q_b|$, $k_b$ for kNV acquisition, and the hot-node ratio $\rho_{CE}$ in TCE.
We also vary the convergence-test threshold, defined as the minimum per-batch improvement in test-sample recall, from 0.1\% to 0.5\%, and the hot-node pruning threshold $t$ from 10 to 20. Throughput changes by less than 3.1\% and 6.3\%, respectively, so the following analysis focuses on the three key parameters.

Figure~\ref{fig:diff} evaluates $|Q_b|$ on T2I-10M, using RoarGraph constructed with 1.0M queries as the baseline.
Increasing $|Q_b|$ from 0.2M to 0.6M improves throughput by 30\%, whereas increasing it from 0.6M to 1.0M brings only an additional 18\% improvement.
CORAL already outperforms RoarGraph with only 0.2M queries.
Figure~\ref{fig:cons_time} reports the corresponding construction time, decomposed into the NAP graph update stage and the CE stage.
As $|Q_b|$ increases, the NAP stage's share of construction cost rises from 24\% to 61\%, while the CE cost remains stable.
RoarGraph's exhaustive CE incurs $3.5\times$ the construction cost of CORAL.

Figure~\ref{fig:diff_k} shows that increasing $k_b$ from 10 to 100 improves throughput by 92\%, while increasing it from 100 to 200 yields only a further 9\% improvement.
Figure~\ref{fig:diff_tce} shows that increasing $\rho_{CE}$ from 10\% to 50\% improves throughput by 87\%, whereas increasing it from 50\% to 100\% reduces throughput by 10\%.
% This is because excessive CE makes the graph less focused on query-relevant regions.

\noindent \textbf{Summary.}
Overall, CORAL is robust across parameter settings.
Its advantage over RoarGraph demonstrates more effective use of query information, while the $|Q_b|$ and $k_b$ results show diminishing returns.
The stable CE cost reflects CORAL's refinement of a controlled subset of structurally critical nodes, keeping refinement cost controllable.
The decline at $\rho_{CE}=100\%$ shows that excessive CE reduces the graph's focus on query-relevant regions.
Together, these trade-offs support $k_b=100$ and $\rho_{CE}=50\%$ as the default settings.
% Larger $|Q_b|$ and $k_b$ improve search efficiency but show diminishing returns, while targeted CE keeps refinement cost controllable.
% The results justify the default choices of $k_b=100$ and $\rho_{CE}=50\%$.

\subsection{Ablation Study (EQ4)}
\label{sec:ablation}

We evaluate the contributions of CAAP, FNAP, and TCE in Table~\ref{tab:alba}.
Removing FNAP and reverting to the conventional single-pivot NAP causes the largest drop, reducing QPS from 10647.30 to 3397.55 on T2I-10M and from 3099.81 to 1703.87 on LAION-10M, which isolates the contribution of the GPU-oriented graph-update design.
Relative to the variant without both CAAP and TCE, adding CAAP alone yields $1.46\times$ and $1.16\times$ speedups on T2I-10M and LAION-10M, respectively, while adding TCE alone yields $2.63\times$ and $1.69\times$ speedups. Full CORAL achieves $3.28\times$ and $1.82\times$ speedups over the same variant.
% CAAP yields $1.46\times$ and $1.16\times$ speedups over the variant without CAAP and TCE on T2I-10M and LAION-10M, respectively.  
% This indicates that improving the utilization efficiency of kNV for queries directly benefits index quality. 
Additional experiments show that as queries are incrementally ingested, CORAL with CAAP efficiently reduces the isolated-node ratio, leaving fewer than 10\% of nodes isolated when the query scale reaches 1M.
% TCE yields $2.63\times$ and $1.69\times$ speedups on the two datasets, respectively. 
% The full system reaches $3.28\times$ and $1.82\times$ speedups, respectively.

\noindent \textbf{Summary.}
Overall, each component improves CORAL's search efficiency.
CAAP improves coverage efficiency by making better use of each query's kNV, FNAP contributes the largest gain through GPU-oriented graph updates, and TCE provides further gains through targeted refinement.
Together, these complementary mechanisms enable the full system to achieve the best performance on both datasets.
% {
% \footnotesize
% \begin{table}[t]
%     \centering
%     \caption{Ablation study.}
    
%     \renewcommand{\arraystretch}{0.8}
%     \setlength{\tabcolsep}{4pt}
    
%     \begin{tabular}{p{2.25cm}cccc}
%     \toprule
%     \multirow{2}{*}{\textbf{Method}} 
%     & \multicolumn{2}{c}{\textbf{T2I-10M}} 
%     & \multicolumn{2}{c}{\textbf{LAION-10M}} \\
%     \cmidrule(lr){2-3} \cmidrule(lr){4-5}
%     & QPS & p95 latency (ms)
%     & QPS & p95 latency (ms)  \\
%     \midrule
    
%     {\small w/o CAAP and TCE} & 3245.73 & 70.14 & 1699.11 & 102.58 \\
%     {\small w/o FNAP} & 3397.55 & 65.33 & 1703.87 & 91.47 \\
%     {\small w/o TCE} & 4746.52 & 37.39 & 1978.34 & 71.27 \\
%     {\small w/o CAAP} & 8536.63 & 28.47 & 2873.68 & 68.68 \\
    
%     \midrule
    
%     {\small \textbf{CORAL}} 
%     & \textbf{10647.3} & \textbf{21.22} 
%     & \textbf{3099.81} & \textbf{54.26} \\
    
%     \bottomrule
%     \end{tabular}
%     \label{tab:alba}
% \end{table}
% }

{
\footnotesize
\begin{table}[t]
    \centering
    \caption{Ablation study.}
    
    \renewcommand{\arraystretch}{0.8}
    \setlength{\tabcolsep}{3.5pt}
    
    \begin{tabular}{p{2.25cm}cccc}
    \toprule
    \multirow{2}{*}{\textbf{Method}} 
    & \multicolumn{2}{c}{\textbf{T2I-10M}} 
    & \multicolumn{2}{c}{\textbf{LAION-10M}} \\
    \cmidrule(lr){2-3} \cmidrule(lr){4-5}
    & QPS & p95 latency (ms)
    & QPS & p95 latency (ms)  \\
    \midrule
    
    {\small w/o CAAP and TCE} & 3245.73 & 70.14 & 1699.11 & 102.58 \\
    {\small w/o FNAP} & 3397.55 & 65.33 & 1703.87 & 91.47 \\
    {\small w/o TCE} & 4746.52 & 37.39 & 1978.34 & 71.27 \\
    {\small w/o CAAP} & 8536.63 & 28.47 & 2873.68 & 68.68 \\
    
    \midrule
    
    {\small \textbf{CORAL}} 
    & \textbf{10647.30} & \textbf{21.22} 
    & \textbf{3099.81} & \textbf{54.26} \\
    
    \bottomrule
    \end{tabular}
    \label{tab:alba}
\end{table}
}
\section{RELATED WORK}
\label{sec:related_work}

\noindent \textbf{OOD ANNS in Cross-Modal Scenarios.}
In cross-modal retrieval, distribution mismatch between query and base vectors degrades ID graph indexes such as HNSW~\cite{malkov2020efficient}, NSG~\cite{fu2019fast}, and DiskANN~\cite{Subramanya2019DiskANN}.
RoarGraph~\cite{chen2024roargraph} guides construction with query modal information, and ParaGraph~\cite{yang2025paragraph} and CDMG~\cite{yue2026efficient} further improve construction efficiency, yet all require costly kNV computation and redundant connectivity enhancement.
Orthogonal to these graph-based indexes, LeanVec~\cite{tepper2025leanvec_paper} addresses OOD search at the representation layer through query-aware projection with quantization.
CORAL instead uses incremental, coverage-aware graph construction to effectively exploit query information at lower construction cost.

\noindent \textbf{Dynamic and Streaming OOD ANNS.}
FreshDiskANN~\cite{singh2021freshdiskann} and SPFresh~\cite{xu2023spfresh} support incremental insertions through buffering and rebalancing, while HNSW~\cite{malkov2020efficient} and DiskANN~\cite{Subramanya2019DiskANN} provide limited update support, and all of these methods target ID settings.
In OOD settings, RoarGraph~\cite{chen2024roargraph} uses localized insertion, limiting global integration, while its tombstone deletion accumulates redundancy and degrades search efficiency.
CORAL instead combines local insertion with global refinement and physically removes deleted nodes through topological repair, preserving connectivity.

\noindent \textbf{GPU-Accelerated OOD ANNS.}
GPU systems such as Faiss~\cite{johnson2021billion} and CAGRA~\cite{ootomo2024cagra} accelerate graph search through parallelism.
ParaGraph~\cite{yang2025paragraph} adapts RoarGraph to GPUs for OOD search, but its pivot-centric construction limits parallelism and caps scalability at 10M vectors.
CORAL instead combines GPU-friendly full-neighborhood projection, targeted refinement, and hierarchical storage for large-scale construction.
\section{CONCLUSION}
\label{sec:conclusion}
We presented CORAL, a GPU-accelerated graph index for cross-modal vector retrieval that combines incremental query-guided construction, CAAP, FNAP, TCE, and topology-aware dynamic updates.
Experiments showed that CORAL improves retrieval effectiveness, reduces construction time by up to 56\% over existing baselines, remains robust under updates, and scales to billions of vectors.
Future work will address evolving data distributions and optimize heterogeneous CPU--GPU--disk execution.

\begin{acks}
% TODO: Replace this placeholder before CRC submission.
This work was supported by the Fundamental and Interdisciplinary Disciplines Breakthrough Plan of the Ministry of Education of China (JYB2025XDXM103) and the Leading Talent of Technological Innovation Program (No. 2023R5214) of Zhejiang Province.
\end{acks}

% \clearpage
\balance

\bibliographystyle{ACM-Reference-Format}
\bibliography{references}

\end{document}